\documentclass[11pt]{article}

\usepackage{ifthen}

\usepackage[utf8]{inputenc}
\usepackage[T1]{fontenc}
\usepackage{lmodern}
\usepackage[DIV=11]{typearea} 

\usepackage{microtype}

\usepackage{amssymb}
\usepackage{amsmath}
\usepackage{amsthm}
\usepackage{thmtools}

\usepackage{graphicx}
\usepackage{xspace}
\usepackage{xcolor}
\usepackage{array}
\usepackage{multirow}
\usepackage{tablefootnote}
\usepackage{caption}
\usepackage{paralist}
\usepackage{enumitem}

\usepackage{comment}

\usepackage[
	style=alphabetic,
	backref=true,
	doi=false,
	url=false,
	maxcitenames=3,
	mincitenames=3,
	maxbibnames=10,
	minbibnames=10,
	backend=bibtex8,
	sortlocale=en_US
]{biblatex}

\usepackage[ocgcolorlinks]{hyperref} 
\usepackage{cleveref}

\newcommand*\samethanks[1][\value{footnote}]{\footnotemark[#1]}

\makeatletter
\g@addto@macro\bfseries{\boldmath}
\makeatother

\makeatletter
\g@addto@macro\mdseries{\unboldmath}
\g@addto@macro\normalfont{\unboldmath}
\g@addto@macro\rmfamily{\unboldmath}
\g@addto@macro\upshape{\unboldmath}
\makeatother

\DeclareCiteCommand{\citem}
    {}
    {\mkbibbrackets{\bibhyperref{\usebibmacro{postnote}}}}
    {\multicitedelim}
    {}

\renewcommand*{\multicitedelim}{\addcomma\space}

\AtEveryCitekey{\clearfield{note}}

\newcommand{\myhref}[1]{%
  \iffieldundef{doi}
    {\iffieldundef{url}
       {#1}
       {\href{\strfield{url}}{#1}}}
    {\href{http://dx.doi.org/\strfield{doi}}{#1}}%
}

\DeclareFieldFormat{title}{\myhref{\mkbibemph{#1}}}
\DeclareFieldFormat
  [article,inbook,incollection,inproceedings,patent,thesis,unpublished]
  {title}{\myhref{\mkbibquote{#1\isdot}}}

\makeatletter
\AtBeginDocument{%
    \newlength{\temp@x}%
    \newlength{\temp@y}%
    \newlength{\temp@w}%
    \newlength{\temp@h}%
    \def\my@coords#1#2#3#4{%
      \setlength{\temp@x}{#1}%
      \setlength{\temp@y}{#2}%
      \setlength{\temp@w}{#3}%
      \setlength{\temp@h}{#4}%
      \adjustlengths{}%
      \my@pdfliteral{\strip@pt\temp@x\space\strip@pt\temp@y\space\strip@pt\temp@w\space\strip@pt\temp@h\space re}}%
    \ifpdf
      \typeout{In PDF mode}%
      \def\my@pdfliteral#1{\pdfliteral page{#1}}
      \def\adjustlengths{}%
    \fi
    \ifxetex
      \def\my@pdfliteral #1{}
      \def\adjustlengths{\setlength{\temp@h}{-\temp@h}\addtolength{\temp@y}{1in}\addtolength{\temp@x}{-1in}}%
    \fi%
    \def\Hy@colorlink#1{%
      \begingroup
        \ifHy@ocgcolorlinks
          \def\Hy@ocgcolor{#1}%
          \my@pdfliteral{q}%
          \my@pdfliteral{7 Tr}
        \else
          \HyColor@UseColor#1%
        \fi
    }%
    \def\Hy@endcolorlink{%
      \ifHy@ocgcolorlinks%
        \my@pdfliteral{/OC/OCPrint BDC}%
        \my@coords{0pt}{0pt}{\pdfpagewidth}{\pdfpageheight}%
        \my@pdfliteral{F}
        \my@pdfliteral{EMC/OC/OCView BDC}%
        \begingroup%
          \expandafter\HyColor@UseColor\Hy@ocgcolor%
          \my@coords{0pt}{0pt}{\pdfpagewidth}{\pdfpageheight}%
          \my@pdfliteral{F}
        \endgroup%
        \my@pdfliteral{EMC}%
        \my@pdfliteral{0 Tr}
        \my@pdfliteral{Q}%
      \fi
      \endgroup
    }%
}
\makeatother

\colorlet{DarkRed}{red!50!black}
\colorlet{DarkGreen}{green!50!black}
\colorlet{DarkBlue}{blue!50!black}

\hypersetup{
	linkcolor = DarkRed,
	citecolor = DarkGreen,
	urlcolor = DarkBlue,
	bookmarks = true,
	bookmarksnumbered = true,
	linktocpage = true
}

\makeatletter
\def\thmt@refnamewithcomma #1#2#3,#4,#5\@nil{%
  \@xa\def\csname\thmt@envname #1utorefname\endcsname{#3}%
  \ifcsname #2refname\endcsname
    \csname #2refname\expandafter\endcsname\expandafter{\thmt@envname}{#3}{#4}%
  \fi
}
\makeatother

\declaretheorem[numberwithin=section,refname={Theorem,Theorems},Refname={Theorem,Theorems}]{theorem}
\declaretheorem[numberlike=theorem,refname={Lemma,Lemmas},Refname={Lemma,Lemmas}]{lemma}
\declaretheorem[numberlike=theorem,refname={Corollary,Corollaries},Refname={Corollary,Corollaries}]{corollary}
\declaretheorem[numberlike=theorem,name={Corollary},refname={Corollary,Corollaries},Refname={Corollary,Corollaries}]{corr}
\declaretheorem[numberlike=theorem,name={Conjecture},refname={Conjecture,Conjectures},Refname={Conjecture,Conjectures}]{conjecture}
\declaretheorem[numberlike=theorem,refname={Proposition,Propositions},Refname={Proposition,Propositions}]{proposition}

\declaretheorem[numberlike=theorem,refname={Claim,Claims},Refname={Claim,Claims}]{claim}
\declaretheorem[numberlike=theorem,refname={Definition,Definitions},Refname={Definition,Definitions}]{definition}

\def\cA{\mathcal{A}}
\def\cB{\mathcal{B}}
\def\cC{\mathcal{C}}

\newcommand{\eps}{\varepsilon}
\renewcommand{\varepsilon}{\epsilon}

\def\poly{\operatorname{poly}}

\global\long\def\oo{\tilde{\tilde{o}}}

\global\long\def\uMv{\textsf{uMv}\xspace}
\global\long\def\guMv{\gamma\mbox{-}\textsf{uMv}\xspace}
\global\long\def\ouMv{\textsf{OuMv}\xspace}
\global\long\def\gouMv{\gamma\mbox{-}\textsf{OuMv}\xspace}
\global\long\def\Mv{\textsf{Mv}\xspace}
\global\long\def\oMv{\textsf{OMv}\xspace}
\global\long\def\goMv{\gamma\mbox{-}\textsf{OMv}\xspace}

\newcommand{\patrascu}{P{\v a}tra{\c s}cu\xspace}
\newcommand{\threesum}{{\sf 3SUM}\xspace}

\title{Unifying and Strengthening Hardness for Dynamic Problems via the Online Matrix-Vector Multiplication Conjecture\thanks{A preliminary version of this paper was presented at the 47th ACM Symposium on Theory of Computing (STOC 2015).}} 
\author{
Monika Henzinger\thanks{Work partially done while visiting the Simons Institute for the Theory of Computing. The research leading to these results has received funding from the European Research Council under the European Union's Seventh Framework Programme (FP7/2007-2013) / ERC Grant Agreement number 340506.} \\{\small University of Vienna, Austria}\\{\small Faculty of Computer Science}
\and
Sebastian Krinninger\samethanks[2] \\{\small University of Vienna, Austria}\\{\small Faculty of Computer Science}
\and
Danupon Nanongkai\thanks{Work partially done while at University of Vienna, Austria.} \\{\small KTH Royal Institute of Technology}
\and
Thatchaphol Saranurak\thanks{Work partially done while at University of Vienna, Austria and Saarland University, Germany.} \\{\small KTH Royal Institute of Technology}
}

\hypersetup{
	pdftitle = {Unifying and Strengthening Hardness for Dynamic Problems via the Online Matrix-Vector Multiplication Conjecture},
	pdfauthor = {Monika Henzinger, Sebastian Krinninger, Danupon Nanongkai, Thatchaphol Saranurak}
}

\date{}

\begin{document}
\maketitle
\begin{abstract}
	Consider the following {\em Online Boolean Matrix-Vector Multiplication} problem: We are given an $n\times n$ matrix $M$ and will receive $n$ column-vectors of size $n$, denoted by $v_1, \ldots, v_n$, one by one. After seeing each vector $v_i$, we have to output the product $Mv_i$ before we can see the next vector. A naive algorithm can solve this problem using $O(n^3)$ time in total, and its running time can be slightly improved to $O(n^3/\log^2 n)$ \citem[Williams SODA'07]{Williams07}. 
	
	We show that a conjecture that there is no {\em truly subcubic} ($O(n^{3-\epsilon})$) time algorithm for this problem can be used to exhibit the underlying polynomial time hardness shared by many dynamic problems. For a number of problems, such as subgraph connectivity, Pagh's problem, $d$-failure connectivity, decremental single-source shortest paths, and decremental transitive closure,	this conjecture implies tight hardness results. Thus, proving or disproving this conjecture will be very interesting as it will either imply several tight unconditional lower bounds or break through a common barrier that blocks progress with these problems. This conjecture might also be considered as strong evidence against any further improvement for these problems since refuting it will imply a major breakthrough for combinatorial Boolean matrix multiplication and other long-standing problems if the term ``combinatorial algorithms'' is interpreted as ``non-Strassen-like algorithms'' \citem[Ballard et al.\ SPAA'11]{BallardDHS12}.
		
	The conjecture also leads to hardness results for problems that were previously based on diverse problems and conjectures -- such as 3SUM, combinatorial Boolean matrix multiplication, triangle detection, and multiphase -- thus providing a uniform way to prove polynomial hardness results for dynamic algorithms; some of the new proofs are also simpler or even become trivial. The conjecture also leads to stronger and new, non-trivial, hardness results, e.g., for the fully dynamic  densest subgraph and diameter problems. 

\end{abstract}
\newpage

\tableofcontents
\newpage

\section{Introduction}\label{sec:intro}

Consider the following problem called {\em Online Boolean Matrix-Vector Multiplication} (\oMv): Initially, an algorithm is given an integer $n$ and an $n\times n$ Boolean matrix $M$. Then, the following protocol repeats for $n$ rounds: At the $i^{\text{th}}$ round, it is given an $n$-dimensional column vector, denoted by $v_i$, and has to compute $Mv_i$. It has to output the resulting column vector before it can proceed to the next round. We want the algorithm to finish the computation as quickly as possible.

This problem is a generalization of the classic {\em Matrix-Vector Multiplication} problem (\Mv), which is the special case with only one vector. The main question is whether we can {\em preprocess} the matrix in order to make the multiplication with $n$ sequentially given
vectors faster than $n$ matrix-vector multiplications. This study dates back to as far as 1955 (e.g., \cite{Motzkin1955}), but most major theoretical work has focused on structured matrices; see, e.g., \cite{Pan01-book,Williams07} for more information. A naive algorithm can multiply the matrix with each vector in $O(n^2)$ time, and thus requires $O(n^3)$ time in total. It was long known that the matrix can be preprocessed in $O(n^2)$ time in order to compute $Mv_i$ in $O(n^2/\log n)$ time, implying an $O(n^3/\log n)$ time algorithm for \oMv; see, e.g., \cite{Savage74,SantoroU86} and a recent extension in \cite{LibertyZ09}. More recently, Williams \cite{Williams07} showed that the matrix can be preprocessed in $O(n^{2+\epsilon})$ time in order to compute $Mv_i$ in $O(n^2/\epsilon\log^2 n)$ time for any $0<\epsilon<1/2$, implying an $O(n^3/\log^2 n)$ time algorithm for \oMv. This is the current best running time for \oMv. In this light, it is natural to conjecture that this problem does not admit a so-called {\em truly subcubic time algorithm}:

\begin{conjecture}[\oMv Conjecture]
	\label{oMv hard}\label{conj:oMv hard} For any constant $\epsilon>0$, there is no $O(n^{3-\epsilon})$-time algorithm that solves \oMv with an error probability of at most $1/3$. 
\end{conjecture}

In fact, it can be argued that this conjecture is implied by the standard {\em combinatorial Boolean matrix multiplication} (BMM) conjecture which states that there is no truly subcubic ($O(n^{3-\eps})$) time {\em combinatorial} algorithm for multiplying two $n\times n$ Boolean matrices if the term ``combinatorial algorithms'' (which is not a well-defined term) is interpreted in a certain way -- in particular if it is interpreted as ``non-Strassen-like algorithms'', as defined in \cite{BallardDHS12}, which captures all known fast matrix multiplication algorithms; see \Cref{sec:discussion} for further discussion. Thus, breaking \Cref{conj:oMv hard} is arguably at least as hard as making a breakthrough for Boolean matrix multiplication and other long-standing open problems (e.g., \cite{DorHZ00,WilliamsW10,AbboudW14,RoddityT11,HenzingerKN13}). 
This conjecture is also supported by an algebraic lower bound \cite{Blaser14}.

\subsection{\oMv-Hardness for Dynamic Algorithms}
We show that the \oMv conjecture can very well capture the underlying polynomial time hardness shared by a large number of dynamic problems, leading to a unification, simplification, and strengthening of previous results. By dynamic algorithm we mean an algorithm that allows a change to the input. It usually allows three operations: (1) {\em preprocessing}, which is called when the input is first received, (2) {\em update}, which is called for every input update, and (3) {\em query}, which is used to request an answer to the problem. For example, in a typical dynamic graph problem, say $s$-$t$ shortest path, we will start with an empty graph at the preprocessing step. Each update operation consists of an insertion or deletion of one edge. The algorithm has to answer a query by returning the distance between $s$ and $t$ at that time. Corresponding to the three operations, we have {\em preprocessing time}, {\em update time}, and {\em query time}. There are two types of bounds on the update time: {\em worst-case bounds}, which bound the time that {\em each individual} update takes in the worst case, and {\em amortized bounds} which bound the time taken by {\em all} updates and then averaging it over all updates. The bounds on query time can be distinguished in the same way. 
We call a dynamic algorithm {\em fully dynamic} if any of its updates can be undone (e.g., an edge insertion can later be undone by an edge deletion); otherwise, we call it {\em partially dynamic}. We call a partially dynamic graph algorithm {\em decremental} if it allows only edge deletions, and {\em incremental} if it allows only edge insertions. For this type of algorithm, the update time is often analyzed in terms of {\em total update time}, which is the total time needed to handle {\em all} insertions or deletions.

Previous hardness results for dynamic problems are often based on diverse conjectures, such as those for \threesum,  combinatorial Boolean matrix multiplication (BMM), triangle detection, all-pairs shortest paths, and multiphase (we provide their definitions in \Cref{sec:conjectures} for completeness). This sometimes made hardness proofs quite intricate since there are many conjectures to start from, which often yield different hardness results, and in some cases none of these results are tight. 
Our approach results in stronger bounds which are tight for some problems. 
Additionally, we show that a number of previous proofs can be unified as they can now start from only one problem, that is \oMv, and can be done in a much simpler way (compare, e.g., the hardness proof for Pagh's problem in this paper and in \cite{AbboudW14}). Thus proving the hardness of a problem via \oMv should be a simpler task.

\begin{figure}
\centering
\includegraphics[width=0.9\textwidth]{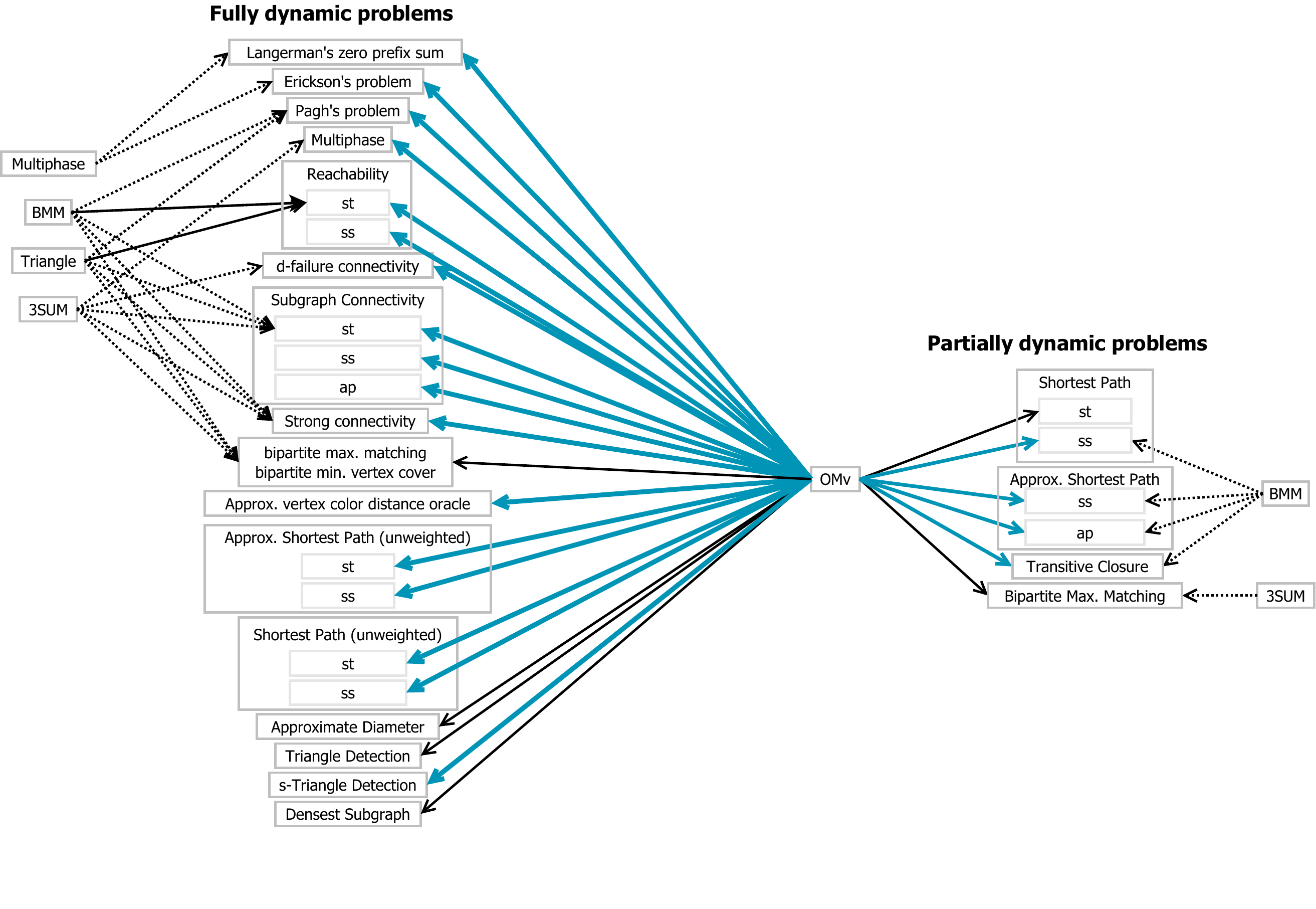} 
\caption{\small
Overview of our and previous hardness results.
An arrow from a conjecture to a problem indicates that there is a hardness result based on the conjecture.
A thick blue arrow indicates that the hardness result is tight, i.e., there is a matching upper bound.
A dotted arrow means that the result is subsumed in our paper. (\Cref{foot:1} discusses results that are not subsumed.) Note that all hardness results based on BMM hold only for combinatorial algorithms.}\label{fig:result summary1}
\end{figure}

We next explain our  main results and the differences to prior work:
As shown in \Cref{fig:result summary1}, we obtain more than 15 new tight\footnote{Our results are tight in one of the following ways: (1) the query time of the existing algorithms cannot be improved without significantly increasing the update time, (2) the update time of the existing algorithms cannot be improved without significantly increasing the query time, (3) the update and query time of the existing algorithms cannot be improved simultaneously, and (4) the approximation guarantee cannot be improved without significantly increasing both query and update time.} hardness results\footnote{\label{foot:1}For the $s$-$t$ reachability problem, our result does not subsume the result based on the Boolean matrix multiplication (BMM) conjecture because the latter result holds only for combinatorial algorithms, and it is in fact larger than an upper bound provided by the non-combinatorial algorithm of Sankowski \cite{Sankowski04} (see \Cref{sec:discussion} for a discussion). 
Also note that the result based on the triangle detection problem which is not subsumed by our result holds only for a more restricted notion of amortization (see \Cref{sec:discussion}). This explains the solid lines in \Cref{fig:result summary1}.}. 
(Details of these results are provided in \Cref{table:tight results,table: tight results partially} for tight results and \Cref{table: summary fully graph,table: summary fully non-graph,table: summary partially}
for improved results. We also provide a summary of the problem definitions in \Cref{table:problem definitions 1,table:problem definitions 2,table:problem definitions 3}.)
(1) Generally speaking, for most previous hardness results in \cite{Patrascu10,AbboudW14,KopelowitzPP14} that rely on various conjectures, except those relying on the Strong Exponential Time Hypothesis (SETH), our \oMv conjecture implies  hardness bounds on the amortized time per operation that are the same or better.
(2) We also obtain new results such as those for vertex color distance oracles (studied in \cite{HermelinLWY11,Chechik12} and used to tackle the minimum Steiner tree problem \cite{LackiOPSZ13}), restricted top trees with edge query problem (used to tackle the minimum cut problem in \cite{FakcharoenpholKNSS14}), and the dynamic densest subgraph problem \cite{BhattacharyaHNT14}.
(3) Some minor improvement can in fact immediately be obtained since our conjecture implies a very strong bound for \patrascu's {\em multiphase} problem \cite{Patrascu10}, giving improved bounds for many problems considered in \cite{Patrascu10}.
We can, however, improve these bounds even more by avoiding a reduction via the multiphase problem. (We discuss this further in \Cref{sec:discussion}.) 
The conjecture leads to an improvement for all problems whose hardness was previously based on \threesum. 
(4) A few other improvements follow from converting previous hardness results that hold only for {\em combinatorial} algorithms into hardness results that hold for {\em any} algorithm. We note that removing the term ``combinatorial'' is an important task as there are algebraic algorithms that can break through some bounds for combinatorial algorithms. (We discuss this more in \Cref{sec:discussion}.)
(5) Interestingly, all our hardness results hold even when we allow an {\em arbitrary polynomial preprocessing time}. This type of results was obtained earlier only in \cite{AbboudW14} for hardness results based on SETH. 
(6) Since the \oMv conjecture can replace all other conjectures except SETH, these two conjectures together are sufficient to show that all hardness results for dynamic problems known so far hold even for arbitrary polynomial preprocessing time.
(7) We also note that all our results hold for a very general type of amortized running time; e.g., they hold even when there is a large (polynomial) number of updates and, for graph problems, even when we start with an empty graph. No previous hardness results, except those obtained via SETH, hold for this case.

We believe that the {\em universality} and {\em simplicity} of the \oMv conjecture will be important not only in proving tight hardness results for well-studied dynamic problems, but also in developing faster algorithms; for example, as mentioned earlier it can be used to show the limits of some specific approaches to attack the minimum Steiner tree and minimum cut problems \cite{LackiOPSZ13,FakcharoenpholKNSS14}. 
Below is a sample of our results. 
A list of all of them and detailed proofs can be found in later sections.

\paragraph{Subgraph Connectivity.} In this problem,  introduced by Frigioni and Italiano~\cite{FrigioniI00},  we are given a graph $G$, and we have to maintain a subset $S$ of nodes where the updates are adding and removing a node of $G$ to and from $S$, which can be viewed as turning nodes on and off. The queries are to determine whether two nodes $s$ and $t$ are in the same connected component in the subgraph induced by $S$.
The best upper bound in terms of $m$ is an algorithm with $\tilde{O}(m^{4/3})$ preprocessing time, $\tilde{O}(m^{2/3})$ amortized update time and $\tilde{O}(m^{1/3})$ worst-case query time \cite{ChanPR11}. There is also an algorithm with $\tilde{O}(m^{6/5})$ preprocessing time, $\tilde{O}(m^{4/5})$ worst-case update time and $\tilde{O}(m^{1/5})$ worst-case query time \cite{Duan10}. An upper bound in terms of $n$ is an algorithm with $\tilde{O}(m)$ preprocessing time, $\tilde{O}(n)$ worst-case update time, and $O(1)$ worst-case query time\footnote{This update time is achieved by using $O(n)$ updates for dynamic connectivity data structure under edge updates by  \cite{KapronKM13}. The query time needs only $O(1)$ time because this data structure internally maintains a spanning forest, so we can label vertices in each component in the spanning forest in time $O(n)$ after each update.}.

For hardness in terms of $m$, Abboud and Vassilevska Williams \cite{AbboudW14} showed that the \threesum conjecture can rule out algorithms with $m^{4/3-\eps}$ preprocessing time, $m^{\alpha-\epsilon}$ amortized update time and $m^{2/3-\alpha-\epsilon}$ amortized query time, for any constants $1/6\leq \alpha\leq 1/3$ and $0<\epsilon<\alpha$. 
In this paper, we show that the \oMv conjecture can rule out algorithms with {\em polynomial} preprocessing time, $m^{\alpha-\epsilon}$ amortized update time and $m^{1-\alpha-\epsilon}$ amortized query time\footnote{We note the following detail: The \threesum-hardness result of Abboud and Vassilevska Williams holds when $m \leq n^{1.5}$ and our hardness result holds when $m\leq \min\{n^{1/\alpha}, n^{1/(1-\alpha)}\}$}, for any $0\leq \alpha\leq 1$. This matches the upper bound of \cite{ChanPR11} when we set $\alpha=2/3$. 

For hardness in terms of $n$, \patrascu \cite{Patrascu10} showed that, assuming the hardness of his {\em multiphase problem}, there is no algorithm with $n^{\delta-\epsilon}$ worst-case update time and query time, for some constant $0<\delta\leq 1$.
By assuming the combinatorial BMM conjecture, Abboud and Vassilevska Williams \cite{AbboudW14} could rule out combinatorial algorithms with $n^{1-\eps}$ amortized update time, and $n^{2-\epsilon}$ query time. 
These two bounds cannot rule out some improvement over \cite{KapronKM13}, e.g., a non-combinatorial algorithm with $n^{1-\eps}$ amortized update time, and $O(1)$ amortized query time. 
In this paper, we show that the \oMv conjecture can rule out {\em any} algorithm with {\em polynomial} preprocessing time, $n^{1-\epsilon}$ amortized update time and $n^{2-\epsilon}$ amortized query time. Thus, there is no algorithm that can improve the upper bound of \cite{KapronKM13} without significantly increasing the query time.

\paragraph{Decremental Shortest Paths.} In the decremental single-source shortest paths problem, we are given an unweighted undirected graph $G$ and a source node $s$. Performing an update means to delete an edge from the graph. A query will ask for the distance from $s$ to some node $v$. 
The best exact algorithm for this problem is due to the classic result of Even and Shiloach \cite{EvenS81} and requires $O(m)$ preprocessing time, $O(mn)$ total update time, and $O(1)$ query time. Very recently, Henzinger, Krinninger, and Nanongkai \cite{HenzingerKNFOCS14} showed a $(1+\epsilon)$-approximation algorithm with $O(m^{1+o(1)})$ preprocessing time, $O(m^{1+o(1)})$ total update time, and $O(1)$ query time.  
Roditty and Zwick \cite{RodittyZESA04} showed that the combinatorial BMM conjecture implies that there is no combinatorial {\em exact} algorithm with 
$mn^{1-\epsilon}$ preprocessing time and $mn^{1-\epsilon}$ total update time if we need $\tilde O(1)$ query time. 
This leaves the open problem whether we can develop a faster exact algorithm for this problem using algebraic techniques (e.g., by adapting Sankowski's techniques \cite{Sankowski04,Sankowski-ESA05,Sankowski-COCOON05}). 
Our \oMv conjecture implies that this is not possible:  there is no exact algorithm with polynomial preprocessing and $mn^{1-\epsilon}$ total update time if we need $\tilde O(1)$ query time.

For the decremental all-pairs shortest paths problem on undirected graphs, $(1+\epsilon)$-approximation algorithms with $\tilde O(mn)$ total update time are also known in both unweighted and weighted cases~\cite{RodittyZ04,Bernstein13}. For combinatorial algorithms, this is tight even in the static setting under the combinatorial BMM conjecture \cite{DorHZ00}. Since fast matrix multiplication can be used to break this bound in the static setting when the graph is dense, the question whether we can do the same in the dynamic setting was raised by Bernstein \cite{Bernstein13}. In this paper, we show that this is impossible under the \oMv conjecture. (Our hardness result holds for any algorithm with approximation ratio less than two.)


\paragraph{Pagh's Problem.} In this problem, we want to maintain a family $X$ of at most $k$ sets $\{X_i\}_{1 \leq i \leq k}$ over $[n]$. An update is by adding the set $X_i \cap X_j$ to $X$. We have to answer a query of the form ``Does element $j$ belong to set $X_i$?''.
A trivial solution to this problem requires $O(kn)$ preprocessing time, $O(n)$ worst-case update time and $O(1)$ worst-case query time.
Previously, Abboud and Williams \cite{AbboudW14} showed that, assuming the combinatorial BMM conjecture, for any $n\leq k\leq n^2$ there is no combinatorial algorithm with $k^{3/2-\epsilon}$ preprocessing time, $k^{1/2-\epsilon}$ amortized update time, and $k^{1/2-\epsilon}$ amortized query time. They also obtained hardness for non-combinatorial algorithms but the bounds are weaker. 
Our \oMv conjecture implies that for any $k=poly(n)$ there is no algorithm with $\poly(k,n)$ preprocessing time, $n^{1-\eps}$ update time, and $k^{1-\epsilon}$ query time, matching the trivial upper bound. Note that our hardness holds against {\em all} algorithms, including non-combinatorial algorithms. Also note that while the previous proof in \cite{AbboudW14} is rather complicated (it needs, e.g., a universal hash function), our proof is almost trivial.

\paragraph{Fully Dynamic Weighted Diameter Approximation.} In this problem, we are give a weighted undirected graph. An update operation adds or deletes a weighted edge. The query asks for the diameter of the graph. 
For the unweighted case, Abboud and Vassilevska Williams \cite{AbboudW14} showed that the Strong Exponential Time Hypothesis (SETH) rules out any $(4/3-\epsilon)$-approximation algorithm with polynomial preprocessing time, $n^{2-\epsilon}$ update and query time. 
Nothing was known for the weighted setting. 
In this paper, we show that for the weighted case, \oMv rules out any $(2-\epsilon)$-approximation algorithm with polynomial preprocessing time, $n^{1/2-\epsilon}$ update time and $n^{1-\epsilon}$ query time. 
This result is among a few that require a rather non-trivial proof.

\subsection{Discussions}\label{sec:discussion}

\paragraph{\oMv vs.\ Combinatorial BMM.} The combinatorial BMM conjecture states that there is no truly subcubic {\em combinatorial} algorithm for multiplying two $n\times n$ Boolean matrices.
There are two important points to discuss here. 
First, it can be easily observed that any reduction from the \oMv problem can be turned into a reduction from the combinatorial BMM problem since, although we get two matrices at once in the BMM problem, we can always pretend that we see one column of the second matrix at a time (this is the \oMv problem). 
This means that bounds obtained via the \oMv conjecture will {\em never} be stronger than bounds obtained via the combinatorial BMM conjecture. However, the latter bounds will hold {\em only for combinatorial algorithms}, leaving the possibility of an improvement via an algebraic algorithm. This possibility cannot be overlooked since there are examples where an algebraic algorithm can break through the combinatorial hardness obtained by assuming the combinatorial BMM conjecture. 
For example, it was shown in \cite{AbboudW14} that the combinatorial BMM conjecture implies that there is no combinatorial algorithm with $n^{3-\epsilon}$ preprocessing time, $n^{2-\epsilon}$ update time, and $n^{2-\epsilon}$ query time for the fully dynamic $s$-$t$ reachability and bipartite perfect matching problems. However, we can break these bounds using Sankowski's algebraic algorithm \cite{Sankowski04,Sankowski07} which requires $n^\omega$ preprocessing time, $n^{1.449}$ worst-case update time, and $O(1)$ worst case query time, where $\omega$ is the exponent of the best known matrix multiplication algorithm (currently, $\omega < 2.3728639$~\cite{Gall14a}).\footnote{Furthermore, the exponent $1.449$ is the result of balancing the terms $n^{1+\epsilon}$ and $ n^{\omega(1, \epsilon, 1)-\epsilon} $, where $\omega(1, \epsilon, 1)$ is the exponent of the best known algorithm~\cite{Gall12} for multiplying an $ n \times n^{\epsilon} $ matrix with an $ n^{\epsilon} \times n $ matrix. The value $1.449$ is obtained by a linear interpolation of the values of $\omega(1, \epsilon, 1)$ reported in~\cite{Gall12}, which upperbounds $\omega(1, \epsilon, 1)$.}

Second, it can be argued that the combinatorial BMM conjecture actually implies the \oMv conjecture, if the term ``combinatorial algorithm'' is interpreted in a certain way. 
Note that while this term has been used very often (e.g., \cite{DorHZ00,WilliamsW10,AbboudW14,RoddityT11,HenzingerKN13}), it is not a well-defined term. Usually it is vaguely used to refer as an algorithm that is different from the ``algebraic'' approach originated by Strassen \cite{Strassen69}; see, e.g., \cite{BlellochVW08,BansalW12,BaschKR95}. 
One formal way to interpret this term is by using the term ``Strassen-like algorithm'', as defined by Ballard~et~al.~\cite{BallardDHS12}. Roughly speaking, a Strassen-like algorithm 
divides {\em both} matrices into constant-size blocks and utilizes an algorithm for multiplying two blocks in order to recursively multiply matrices of arbitrary size (see \cite[Section 5.1]{BallardDHS12} for a detailed definition\footnote{Note that Ballard~et~al.\ also need to include a technical assumption in the Strassen-like algorithms that they consider to prove their results (see \cite[Section 5.1.1]{BallardDHS12}). This assumption is irrelevant to us.}). As pointed out in \cite{BallardDHS12}, this is the structure of all the fast matrix multiplication algorithms that were obtained since Strassen's, including the recent breakthroughs by Stothers~\cite{Stothers2010} and Vassilevska Williams~\cite{VWilliams12}.
Since \oMv reveals one column of the second matrix at a time, it naturally disallows an algorithm to utilize block multiplications, and thus Strassen-like algorithms cannot be used to solve \oMv. 

We note that the \oMv problem actually excludes even some combinatorial BMM algorithms; e.g., the  $O(n^3 (\log\log n)^2 / \log^{2.25}n)$-time algorithm of Bansal and Williams \cite{BansalW12} cannot be used to solve \oMv.
Finally, even if one wants to interpret the term ``combinatorial algorithm'' differently and argue that the combinatorial BMM conjecture does not imply the \oMv conjecture, we believe that breaking the \oMv conjecture will still be a breakthrough since it will yield a fast matrix multiplication algorithm that is substantially different from those using Strassen's approach.

\paragraph{\oMv vs.\ Multiphase.} \patrascu \cite{Patrascu10} introduced a dynamic version of set disjointness called {\em multiphase problem}, which can be rephrased as a variation of the Matrix-Vector multiplication problem as follows (see \cite{Patrascu10} for the original definition). Let $k$, $n$, and $\tau$ be some parameters. First, we are given a $k\times n$ Boolean matrix $M$ and have $O(nk\cdot \tau)$ time to preprocess $M$. Second, we are given an $n$-dimensional vector $v$ and have $O(n\cdot \tau)$ additional computation time. Finally, we are given an integer $1\leq i\leq n$ and must output $(Mv)[i]$ in $O(\tau)$ time. \patrascu conjectured that if there are constants $\gamma>0$ and $\delta>0$ such that $k=n^\gamma$, then any solution to the multiphase problem in the Word RAM model requires $\tau=n^\delta$, and used this conjecture to prove polynomial time hardness for several dynamic problems. How strong these hardness bounds are depends on how hard one believes the multiphase problem to be. 
By a trivial reduction, the \oMv conjecture implies that the multiphase conjecture holds with $\delta=1$ when $\gamma=1$.  
(We found it quite surprising that viewing the multiphase problem as a matrix problem, instead of a set problem as originally stated, can give an intuitive explanation for a possible value of $\delta$.)
This implies the strongest bound possible for the multiphase problem.
Moreover, while hardness based on the multiphase problem can only hold for a {\em worst-case} time bound, it can be shown that under a general condition we can make them hold for an {\em amortized} time bound too if we instead assume the \oMv conjecture (see \Cref{sec:multiphase} for details).
Thus, with the \oMv conjecture it seems that we do not need the multiphase conjecture anymore.
Note that, as argued before, we can also conclude that the combinatorial BMM conjecture implies the multiphase conjecture. To the best of our knowledge this is the first connection between these conjectures.

\paragraph{\oMv vs.\ 3SUM and SETH.} As mentioned earlier, all previous hardness results that were based on the \threesum conjecture can be strengthened through the \oMv conjecture. However, we do not have a general mechanism that can always convert any hardness proof based on \threesum into a proof based on \oMv. Finding such a mechanism would be interesting. 

Techniques for proving hardness for dynamic algorithms based on SETH were very recently introduced in \cite{AbboudW14}. Results from these techniques are the only ones that cannot be obtained through \oMv. SETH together with \oMv seems to be enough to prove all the hardness results known to date. It would be very interesting if the number of conjectures one has to start with can be reduced to one.

\paragraph{Remark On the Notion of Amortization.} We emphasize that there are two different ways to define the notion of amortized update time. First, we can define it as an amortized update time when we {\em start from an empty graph}; equivalently, the update time has to be amortized over all edges that ever appear. 
The second way is to allow the algorithm to preprocess an arbitrary input graph and amortize the update time over all  updates (not counting the edges in the initial graph as updates); for example, one can start from a graph with $n^2$ edges and have only $n$ updates. 
The first definition is more common in the analysis of dynamic algorithms but it is harder to prove hardness results for it. Our hardness results hold for this type of amortization; in fact, they hold even when there is a large (polynomial) number of updates. Many previous hardness results hold only for the second type of amortization, e.g., the results in \cite{AbboudW14} that are not based on SETH and 3SUM.

\subsection{Notation} 
All matrices and vectors in this paper are Boolean and $\tilde{O}(\cdot)$ hides logarithmic factors in $O(\cdot)$.
We will also use the following non-standard notation.

\begin{definition}
[$\oo$ Notation] \label{oo def} 
For any parameters $n_1,n_2,n_3$, we say that a function $f(n_{1},n_{2},n_{3})=\oo(n_{1}^{c_{1}}n_{2}^{c_{2}}n_{3}^{c_{3}})$
iff there exists some constant $\epsilon>0$ such that $f=O(n_{1}^{c_{1}-\epsilon}n_{2}^{c_{2}}n_{3}^{c_{3}}+n_{1}^{c_{1}}n_{2}^{c_{2}-\epsilon}n_{3}^{c_{3}}+n_{1}^{c_{1}}n_{2}^{c_{2}}n_{3}^{c_{3}-\epsilon})$.
We use the analogous definition for functions with one or two parameters.
\end{definition}

\subsection{Organization}

Instead of starting from \oMv, our reductions will start from an intermediate problem called \ouMv. We describe this and prove necessary results in \Cref{sec:intermediate}. 
In \Cref{sec:full:fully hardness} we prove hardness results for the amortized update time of fully dynamic algorithms and the worst-case update time of partially dynamic algorithms. In \Cref{sec:full:partially hardness} we prove hardness results for the total update time of partially dynamic problems. In \Cref{sec:futher discussions} we provide further discussions.

\begin{table}
\footnotesize
\begin{tabular}{|>{\centering}m{0.15\textwidth}|>{\centering}p{0.1\textwidth}|>{\centering}p{0.25\textwidth}|>{\centering}m{0.2\textwidth}|>{\centering}p{0.25\textwidth}|}
\hline 
\multicolumn{2}{|c|}{Name and short name} & Input & Update & Query\tabularnewline
\hline 
\hline 
\multirow{3}{0.15\textwidth}{Subgraph Connectivity} & st-SubConn & \multirow{3}{0.25\textwidth}{A fixed undirected graph $G$, and a subset $S$ of its vertices} & \multirow{3}{0.2\textwidth}{Insert/remove a node into/from $S$} & Are $s$ and $t$ connected in $G[S]$ the subgraph induced by $S$?\tabularnewline
\cline{2-2} \cline{5-5} 
 & ss-SubConn &  &  & For any $v$, are $s$ and $v$ connected in $G[S]$ the subgraph
induced by $S$?\tabularnewline
\cline{2-2} \cline{5-5} 
 & ap-SubConn &  &  & For any $u,v$, are $u$ and $v$ connected in $G[S]$ the subgraph
induced by $S$?\tabularnewline
\hline 
\multirow{3}{0.15\textwidth}{Reachability} & st-Reach & \multirow{3}{0.25\textwidth}{A directed graph} & \multirow{3}{0.2\textwidth}{Edge insertions/deletions} & Is $t$ reachable from $s$?\tabularnewline
\cline{2-2} \cline{5-5} 
 & ss-Reach &  &  & For any $v$, is $v$ reachable from $s$?\tabularnewline
\cline{2-2} \cline{5-5} 
 & ap-Reach (Transitive Closure or TC) &  &  & For any $u,v$, is $v$ reachable from $u$? \tabularnewline
\hline 
\multirow{3}{0.15\textwidth}{Shortest Path (undirected)} & st-SP & \multirow{3}{0.25\textwidth}{An undirected unweighted graph} & \multirow{3}{0.2\textwidth}{Edge insertions/deletions} & Find the distance $d(s,t)$.\tabularnewline
\cline{2-2} \cline{5-5} 
 & ss-SP &  &  & Find the distance $d(s,v)$, for any $v$.\tabularnewline
\cline{2-2} \cline{5-5} 
 & ap-SP &  &  & Find the distance $d(u,v)$, for any $u,v$.\tabularnewline
\hline 
\multicolumn{2}{|c|}{Triangle Detection} & \multirow{2}{0.25\textwidth}{An undirected graph } & \multirow{2}{0.2\textwidth}{Edge insertions/deletions} & Is there a  triangle in the graph?\tabularnewline
\cline{1-2} \cline{5-5} 
\multicolumn{2}{|c|}{s-Triangle Detection} &  &  & Is there a triangle containing $s$ in the graph?\tabularnewline
\hline 
\end{tabular}

\caption{Definitions of dynamic graph problems (1)}\label{table:problem definitions 1}
\end{table}

\begin{table}
\footnotesize
\begin{tabular}{|>{\centering}m{0.15\textwidth}|>{\centering}p{0.25\textwidth}|>{\centering}m{0.25\textwidth}|>{\centering}p{0.30\textwidth}|}
\hline 
Name  & Input & Update & Query\tabularnewline
\hline 
\hline 
Densest Subgraph &   An undirected graph  & Edge insertions/deletions & What is the density $|E(S)|/|S|$ of the densest subgraph $S$?\tabularnewline
\hline 
$d$-failure connectivity \cite{DuanP10}  &   A fixed undirected graph  & Roll back to original graph. Then remove any $d$ vertices from the
graph. & Is $s$ connected to $t$, for any given $(s,t)$?\tabularnewline
\hline 
Vertex color distance oracle \cite{Chechik12,LackiOPSZ13} &   A fixed undirected graph  & Change the color of any vertex & Given $s$ and color $c$, find the shortest distance from $s$ to
any $c$-colored vertex.\tabularnewline
\hline 
Diameter &   An undirected graph  & Edge insertions/deletions & Find the diameter of the graph.\tabularnewline
\hline 
Strong Connectivity & A directed graph & Edge insertions/deletions & Is the graph strongly connected?\tabularnewline
\hline 
\end{tabular}

\caption{Definitions of dynamic graph problems (2)}\label{table:problem definitions 2}
\end{table}

\begin{table}
\footnotesize
\begin{tabular}{|>{\centering}m{0.15\textwidth}|>{\centering}p{0.25\textwidth}|>{\centering}p{0.25\textwidth}|>{\centering}p{0.30\textwidth}|}
\hline 
Name  & Input & Update & Query\tabularnewline
\hline 
\hline 
Pagh's Problem  &   A family $\mathcal{X}$ of sets $X_{1},X_{2},\dots,X_{k}\subseteq[n]$ & Given $i,j$, insert $X_{i}\cap X_{j}$ into $\mathcal{X}$ & Given index $i$ and $j\in[n]$, is $j\in X_{i}$?\tabularnewline
\hline 
Langerman's Zero Prefix Sum &   An array $A[1\dots n]$ of integers & Set $A[i]=x$ for any $i$ and $x$ & Is there a $k$ s.t. $\sum_{i=1}^{k}A[i]=0$?\tabularnewline
\hline 
Erickson's Problem &   A matrix of integers of size $n\times n$ & Increment all values in a specified row or column & Find the maximum value in the matrix\tabularnewline
\hline 
\end{tabular}

\caption{Definitions of dynamic non-graph problems}\label{table:problem definitions 3}

\end{table}

\section{Intermediate Problems}\label{sec:intermediate}

In this section we show that the \oMv conjecture implies that \oMv is hard even when there is a polynomial preprocessing time and different dimension parameters (\Cref{sec:stronger conjectures}). Then in \Cref{sec:full hardness ouMv}, we present the problem whose hardness can be proved assuming the \oMv conjecture, namely the {\em online vector-matrix-vector multiplication} (\ouMv) problem, which is the key starting points for our reductions in later sections.

\subsection{\oMv with Polynomial Preprocessing Time and Arbitrary Dimensions}\label{sec:stronger conjectures}

We first define a more general version of the $\oMv$ problem: (1) we allow the algorithm to preprocess the matrix before the vectors arrive and (2) we allow the matrix to have arbitrary dimensions with a promise that the size of minimum dimension is not too ``small'' compared to the size of maximum dimension.

\begin{definition}
	[$\goMv$]
	Let $\gamma>0$ be a fixed constant. 
	An algorithm for the $\goMv$ problem is given \emph{parameters} $n_1,n_2,n_3$ as input with a promise that 
	$n_1 = \lfloor n_2^\gamma \rfloor$.
	Next, it is given a matrix $M$ of size $n_{1}\times n_{2}$ that can be preprocessed.
	Let $p(n_{1},n_{2})$ denote the \emph{preprocessing time}.
	After the preprocessing, an online sequence of vectors $v^{1},\dots,v^{n_{3}}$ is presented one after the other and the task is to compute each $Mv^{t}$ before $v^{t+1}$ arrives.
	Let $c(n_{1},n_{2},n_{3})$ denote the \emph{computation time} over the whole sequence.
\end{definition}
Note that the $\goMv$ problem can be trivially solved with $O(n_1n_2n_3)$ total computing time and without preprocessing time. 
Obviously, the \oMv conjecture implies that this running time is (almost) tight when $ n_1 = n_2 = n_3 = n $. Interestingly, it also implies that this running time is tight for other values of $ n_1$, $n_2$, and $n_3$:

\begin{theorem}
	\label{general oMv hard}
	For any constant $\gamma>0$,
	\Cref{oMv hard} implies that there is no algorithm for $\goMv$ with parameters $n_1,n_2,n_3$ using 
	preprocessing time $p(n_{1},n_{2})=poly(n_{1},n_{2})$ 
	and computation time $c(n_{1},n_{2},n_{3})=\oo(n_{1}n_{2}n_{3})$ that has an error probability of at most $1/3$.
\end{theorem}

The rest of this section is devoted to proving the above theorem. 
The proof proceeds in two steps. 
First, we show that, assuming \Cref{oMv hard}, there is no algorithm for $\goMv$ 
when the preprocessing time is $(n_{1}n_{2})^{1+\epsilon}$ for any constant $\epsilon<1/2$ and 
the computation time is $\oo(n_1 n_2 n_3)$.

\begin{lemma}
\label{oMv<=oMv n*n}
For any constant $\gamma>0$ and integer $n$, fix any $n_1,n_2$ where $n_1=n_2^\gamma$, $\max\{n_1,n_2\}=n$ and $n_3=n$.\footnote{
	Actually, we need $n_1 = \lfloor n_2^\gamma \rfloor$ 
	but from now we will always omit it and assume that $n_2^\gamma$ is an integer.
	This affects the running time of the statement by at most a constant factor.
}
Suppose there is an algorithm $\cA$ for $\goMv$ with parameters $n_1,n_2,n_3$ using 
preprocessing time $p(n_{1},n_{2})$ and
computation time $c(n_{1},n_{2},n_3)$ that has an error probability of at most $\delta$. 
Then there is an algorithm $\cB$ for $\oMv$ with parameter $n$ using
(no preprocessing time and) computation time $\tilde{O}(\frac{n^{2}}{n_{1}n_{2}}p(n_{1},n_{2})+\frac{n^{2}}{n_{1}n_{2}}c(n_{1},n_{2},n)+\frac{n^{3}}{n_{2}})$ that has an error probability of at most $\delta$.
\end{lemma}

\begin{proof}
We will construct $\cB$ by using $\cA$ as a subroutine. 
We partition $M$ into blocks $\{M_{x,y}\}_{1\le x\le n/n_{1},1\le y\le n/n_{2}}$
where $M_{x,y}$ is of size $n_{1}\times n_{2}$.\footnote{
	Here we assume that $n_1$ and $n_2$ divides $n$ and we will similarly assume this whenever we divide a matrix into a blocks.
	This assumption can be removed easily: for each ``boundary'' blocks $M_{x,y}$ where $x = \lceil n/n_1 \rceil$ or $y = \lceil n/n_2 \rceil$, we keep the size $M_{x,y}$ to be $n_1 \times n_2$ but it may overlap with other block. 
	This will affect the running time by at most a constant factor.
} 
We feed $M_{x,y}$ to an instance $I_{x,y}$ of $\cA$ and preprocess 
using $\frac{n^{2}}{n_{1}n_{2}}p(n_{1},n_{2})$
time.
For each vector $v^t$, we partition it into blocks $\{v_{y}^{t}\}_{1\le y\le n/n_{2}}$ each
of size $n_{2}$. For each $x,y$, we compute $M_{x,y}v_{y}^{t}$
using the instance $I_{x,y}$ for all $t\le n$.
The total time for computing $M_{x,y}v_{y}^{t}$, for all $x,y,t$, is 
$\frac{n^{2}}{n_{1}n_{2}}c(n_{1},n_{2},n)$. 
We keep the error probability to remain at most $\delta$ by a standard application of the Chernoff bound:
repeat the above procedure for computing $M_{x,y}v_{y}^{t}$, for each $x,y,t$, 
$O(\log(n_1 n_2 n_3))$ many times and take the most frequent answer.

Let $c^{t}=Mv^{t}$. We write $c^{t}$ as blocks $\{c_{x}^{t}\}_{1\le x\le n/n_{1}}$
each of size $n_{1}$. Since $c_{x}^{t}=\bigvee_{y}M_{x,y}v_{y}^{t}$ (bit-wise OR)
for each $x,t$, we can compute $c_{x}^{t}$ in time $\frac{n}{n_{2}}n_{1}$
for each $x,t$ (there are $\frac{n}{n_2}$ many $y$'s and $M_{x,y}v_{y}^{t}$ is a vector of size $n_1$). The total time for this, over all $x,t$, is $\frac{n}{n_{1}}n\times\frac{n}{n_{2}}n_{1}=\frac{n^{3}}{n_{2}}$.
\end{proof}

\begin{corollary}
\label{oMv hard small preproc}
For any constants $\gamma>0$ and $\epsilon<1/2$, \Cref{oMv hard} implies that there is no algorithm for $\goMv$ with parameters $n_1,n_2,n_3$
using preprocessing time $p(n_{1},n_{2}) \le (n_{1}n_{2})^{1+\epsilon}$ and
computation time $c(n_{1},n_{2},n_{3})=\oo(n_{1}n_{2}n_{3})$ that has an error probability of at most $1/3$.
\end{corollary}

\begin{proof}
Suppose there is such an algorithm. Then by \Cref{oMv<=oMv n*n},
we can solve $\oMv$ with parameter $n$
in time $\tilde{O}( \frac{n^{2}}{n_{1}n_{2}}p(n_{1},n_{2})+\frac{n^{2}}{n_{1}n_{2}}c(n_{1},n_{2},n)+\frac{n^{3}}{n_{2}} )$
with error probability at most $1/3$ where $n_1=n_2^\gamma$ and $\max\{n_1,n_2\}=n$. 
We have $\frac{n^{2}}{n_{1}n_{2}}p(n_{1},n_{2}) \le n^{2}(n_{1}n_{2})^{\epsilon} \le n^{2+2\epsilon} = \oo(n^3)$,
and $\frac{n^{2}}{n_{1}n_{2}}c(n_{1},n_{2},n)+\frac{n^{3}}{n_{2}}=\oo(n^{3})$
where the last equality holds because $n_{1},n_{2}\ge \min\{ n^{\gamma},n^{1/\gamma} \}$.
The total time is $\oo(n^{3})$ and the error probability is at most $1/3$ contradicting \Cref{oMv hard}.
\end{proof}

For the second step, we show that the hardness of $\goMv$ even when the preprocessing time is $p(n_{1},n_{2})=\poly(n_{1},n_{2})$.

\begin{lemma}
\label{large preproc oMv<=oMv}
For any constant $\gamma>0$ and integers $n_1,n_2,n_3$ where $n_1 = n_2^\gamma$,
fix any $k_{1},k_{2}$ where $k_{1}\le n_1, k_2 \le n_2$, and $k_1=k_2^\gamma$.
Suppose there is an algorithm $\cA$ for $\goMv$ with parameters $k_1,k_2,n_3$,
preprocessing time $p(k_{1},k_{2})$ and computation time $c(k_{1},k_{2},n_{3})$ that has an error probability of at most $\delta$.
Then there is an algorithm $\cB$ for $\goMv$ with parameters $n_1,n_2,n_3$,
preprocessing time $\tilde{O}( \frac{n_{1}n_{2}}{k_{1}k_{2}}p(k_{1},k_{2}) )$ and
computation time $\tilde{O}( \frac{n_{1}n_{2}}{k_{1}k_{2}}c(k_{1},k_{2},n_{3})+\frac{n_{1}n_{2}n_{3}}{k_{2}} )$ that has an error probability of at most $\delta$.
\end{lemma}
\begin{proof}
We will construct $\cB$ by using $\cA$ as a subroutine.
We partition $M$ into blocks $\{M_{x,y}\}_{1\le x\le n_{1}/k_{1},1\le y\le n_{2}/k_{2}}$
where $M_{x,y}$ is of size $k_{1}\times k_{2}$. We feed $M_{x,y}$
to an instance $I_{x,y}$ of $\cA$ and then preprocess. 
This takes $\frac{n_{1}n_{2}}{k_{1}k_{2}}p(k_{1},k_{2})$ total preprocessing time.

Once the vector $v^{t}$ arrives, we partition it into blocks $\{v_{y}^{t}\}_{1\le y\le n_{2}/k_{2}}$
each of size $k_{2}$. For each $x,y$, we compute $M_{x,y}v_{y}^{t}$
using the instance $I_{x,y}$. The total computation time, over all $t$, for
doing this will be $\frac{n_{1}n_{2}}{k_{1}k_{2}}c(k_{1},k_{2},n_{3})$.
By repeating the procedure for a logarithmic number of times as in \Cref{oMv<=oMv n*n}, the error probability remains at most $\delta$. 

Let $c^{t}=Mv^{t}$. We write $c^{t}$ as blocks $\{c_{x}^{t}\}_{1\le x\le n_{1}/k_{1}}$
each of size $k_{1}$. Since $c_{x}^{t}=\bigvee_{y}M_{x,y}v_{y}^{t}$
for each $x,t$, we can compute $c_{x}^{t}$ in time $\frac{n_{2}}{k_{2}}k_{1}$
for each $x,t$. The total time for this, over all $x,t$, is $\frac{n_{1}}{k_{1}}n_{3}\times\frac{n_{2}}{k_{2}}k_{1}=\frac{n_{1}n_{2}n_{3}}{k_{2}}$.\end{proof}

We conclude the proof of \Cref{general oMv hard}.

\begin{proof}[Proof of \Cref{general oMv hard}]
We construct an algorithm $\cB$ for $\goMv$ with parameters $n_1,n_2,n_3$ that contradicts \Cref{oMv hard} 
by using an algorithm $\cA$ from the statement of \Cref{general oMv hard} as a subroutine. 
That is, $\cA$ is an algorithm for $\goMv$ with parameters $k_1,k_2,k_3$,
preprocessing time $(k_{1}k_{2})^{c}$ for some constant $c$,
and computation time $\oo(k_{1}k_{2}k_{3})$ with error probability $1/3$.
Let $\epsilon < \min\{1/2,c\}$ be a constant.
We choose $k_1,k_2$ such that $k_{1}=n_{1}^{\epsilon/c}$, $k_{2}=n_{2}^{\epsilon/c}$ and $k_3 = n_3$.

Note that $k_1 \le n_1$, $k_2 \le n_2$ and $k_1 = k_2^\gamma$. So we can apply \Cref{large preproc oMv<=oMv}
and get $\cB$ which has error probability $1/3$ and, ignoring polylogarithmic factors, uses
preprocessing time $\frac{n_{1}n_{2}}{k_{1}k_{2}}p(k_{1},k_{2})\le n_{1}n_{2}\cdot p(k_{1},k_{2})\le(n_{1}n_{2})^{1+\epsilon}$
and computation time $\frac{n_{1}n_{2}}{k_{1}k_{2}}c(k_{1},k_{2},n_{3})+\frac{n_{1}n_{2}n_{3}}{k_{2}}=\frac{n_{1}n_{2}}{n_{1}^{\epsilon/c}n_{2}^{\epsilon/c}}\oo(n_{1}^{\epsilon/c}n_{2}^{\epsilon/c}n_{3})+\frac{n_{1}n_{2}n_{3}}{n_{2}^{\epsilon/c}}=\oo(n_{1}n_{2}n_{3})$.
This contradicts \Cref{oMv hard} by \Cref{oMv hard small preproc}.
\end{proof}

\subsection{The Online Vector-Matrix-Vector Multiplication Problem (\ouMv)}\label{sec:full hardness ouMv}

Although we base our results on the hardness of \oMv, the starting point of most of our reductions is a slightly different problem called online vector-matrix-vector multiplication problem. In this problem, we multiply the matrix with two vectors, one from the left and one from the right.

\begin{definition}
	[$\gouMv$ problem]
	Let $\gamma>0$ be a fixed constant.
	An algorithm for the $\gouMv$ problem is given \emph{parameters} $n_1,n_2,n_3$ as its input with the promise that 
	$n_1 = \lfloor n_2^\gamma \rfloor$.
	Next, it is given a matrix $M$ of size $n_{1}\times n_{2}$ that can be preprocessed.
	Let $p(n_{1},n_{2})$ denote the \emph{preprocessing time}.
	After the preprocessing, an online sequence of vector pairs $(u^{1},v^{1}),\dots,(u^{n_{3}},v^{n_{3}})$ is presented one after the other and the task is to compute each $(u^{t})^{\top}Mv^{t}$ before $(u^{t+1},v^{t+1})$ arrives.
	Let $c(n_{1},n_{2},n_{3})$ denote the \emph{computation time} over the whole sequence.
	The $\guMv$ problem with parameters $n_1,n_2$ is the special case of $\gouMv$ where $ n_3 = 1 $.
\end{definition}
We also write $\ouMv$ and $\uMv$ to refer to, respectively, $\gouMv$ and $\guMv$ without the promise.
Our reductions will exploit the fact that the result of this multiplication is either $ 0 $ or $ 1 $; thus using only $ 1 $ bit as opposed to $ n $ bits in $ \oMv $.
Starting from $\ouMv$ instead of $\oMv$ will thus give simpler reductions and better lower bounds on the query time.
Using a technique for finding ``witnesses'', which will be defined below, when the result of a vector-matrix-vector multiplication is $ 1 $, we can reduce the  $ \goMv $ problem to the $ \gouMv $ problem and establish the following hardness for $ \gouMv $. 

\begin{theorem}
	\label{ouMv hard}
	For any constant $\gamma>0$,
	\Cref{oMv hard} implies that there is no algorithm for $\gouMv$ with parameters $n_1,n_2,n_3$ using
	preprocessing time $ p(n_1, n_2) = poly(n_{1},n_{2})$ and
	computation time $c(n_{1},n_{2},n_{3})=\oo(n_{1}n_{2}n_{3})$ that has an error probability of at most $1/3$.
\end{theorem}

An $\guMv$ algorithm with preprocessing time $p(n_1,n_2)$ and computation time $c(n_1,n_2)$ implies
an $\gouMv$ algorithm with preprocessing time $\tilde{O}(p(n_1,n_2))$, computation time $\tilde{O}(n_3 c(n_1,n_2))$ and the same error probability by a standard application of the Chernoff bound as in the proof of \Cref{oMv<=oMv n*n}.
Therefore, we also get the following:
\begin{corollary}
	\label{uMv hard}
	For any constant $\gamma>0$,
	\Cref{oMv hard} implies that there is no algorithm for $\guMv$ with parameters $n_1,n_2$
	using preprocessing time $p(n_1,n_2)=\poly(n_1,n_2)$ and
	computation time $c(n_1,n_2)=\oo(n_1 n_2)$ that has an error probability of at most $1/3$.
\end{corollary}

The rest of this section is devoted to the proof of \Cref{ouMv hard}.

\begin{definition}
[Witness of $\ouMv$]
We say that any index $i$ is a {\em witness} for a pair of vectors  $(u^{t},v^{t})$ in an instance $I$
of $\ouMv$ if $u_{i}^{t}\wedge(Mv^{t})_{i}=1$, i.e., the $i$-th entries of vectors $u^t$ and $Mv^{t}$ are both one.
\end{definition}

Observe that $(u^{t})^{\top}Mv^{t}=1$ if and only if there is a witness for $(u^{t},v^{t})$.  
The problem of \emph{listing all the witnesses of $\gouMv$} is defined similarly as $\gouMv$ except that
for each vector pair $(u^{t},v^{t})$, we have to list all witnesses of $(u^{t},v^{t})$ 
(i.e., output every index $i$ such that $i$ is a witness)
before $(u^{t+1},v^{t+1})$ arrives.
We first show a reduction from the problem of listing all the witnesses of $\gouMv$ to the $\gouMv$ problem itself.
The reduction is similar to \cite[Lemma 3.2]{WilliamsW10}.

\begin{lemma}
\label{list witness}
Fix any constant $\gamma>0$ and integers $n_1,n_2$ and $n_3$. 
Suppose there is an algorithm $\cA$ for $\gouMv$ with parameters $n_1,n_2,n_3$
preprocessing time $p(n_{1},n_{2})$ and computation time $c(n_{1},n_{2},n_{3})$ that has an error probability of at most $\delta$.
Then there is an algorithm~$\cB$ for
listing all witnesses of $\gouMv$ with the same parameters using
preprocessing time $\tilde{O}(p(n_{1},n_{2}))$ and
computation time $\tilde{O}((1 + \sum_{t}\tilde{O}(w_{t}) / n_3) \cdot c(n_{1},n_{2},n_3))$ that has an error probability of at most $\delta$, 
where $w_{t}$ is the number of witnesses of $(u^{t},v^{t})$.\end{lemma}
\begin{proof}
We will show a reduction for deterministic algorithms.
This reduction can be extended to work for randomized algorithms as well by a standard application of the Chernoff bound as in the proof of \Cref{oMv<=oMv n*n}. 
We construct $\cB$ using $\cA$ as a subroutine.
Let $M$ be the input matrix of~$\cB$. We use $\cA$ to preprocess $M$.

For any vector pair $(u,v)$, we say the ``\emph{query} $(u,v)$'' to $\cA$ returns true if, by using $\cA$, we get $u^\top Mv =1$.
For any a set of indices $I$ of entries of $u$, 
let $u_I$ be the $n_1$-dimensional vector where $(u_I)_i = (u)_i$ for all $i\in I$ and $(u_I)_i = 0$ otherwise. 
Suppose that $I$ contains $w$ many witnesses of $(u,v)$.
We now describe a method to identify all witnesses of $(u,v)$ in $I$ using $1+O(w\log n_1)$ queries. 
Note that the witnesses of $(u,v)$ contained in $I$ are exactly the witnesses of $(u_I,v)$.
We check if $w=0$ by querying $(u_I,v)$ one time.
If $w>0$, then we identify an arbitrary witness of $(u_I,v)$ one by one using binary search.
More precisely, if $I$ is of size one, then return the only index $i\in I$ which must be a witness.
Otherwise, let $I_0 \subset I$ be a set of size $\lfloor |I|/2 \rfloor$. 
If the query $(u_{I_0},v)$ return true, then recurse on $I_0$. Otherwise, recurse on $I\setminus I_0$. 
This takes $O(\log n_1)$ queries because $|I| \le n_1$.
Once we find a witness $i\in I$, we do the same procedure on $I\setminus\{i\}$ until we find all $w$ witnesses. 
Therefore, the total number of queries for finding $w$ many witnesses of $(u,v)$ in~$I$ is $1+O(w\log n_1)$.

Once $(u^t,v^t)$ arrives, we list all witnesses of $(u^t,v^t)$ using $1+O(w_t\log n_1)$ queries by the above procedure where $I = [n_1]$. The total number of queries is $n_{3}+\sum_{t}\tilde{O}(w_{t})$. 
However, $\cA$ is an algorithm for $\gouMv$ with parameters $n_1,n_2,n_3$.
So once there are $n_3$ queries to $\cA$, we need to roll back $\cA$ to the state right after preprocessing.
Hence, we need to roll back $\frac{n_{3}+\sum_{t}\tilde{O}(w_{t})}{n_3} = 1 + \sum_{t}\tilde{O}(w_{t}) / n_3$ times. 
Therefore, the total computation time is $(1 + \sum_{t}\tilde{O}(w_{t}) / n_3) c(n_{1},n_{2},n_3)$.
\end{proof}

Next, using \Cref{list witness}, we can show the reduction from  $\goMv$ to $\gouMv$.

\begin{lemma}
\label{oMv <= ouMv}
For any constant $\gamma>0$ and integers $n_1,n_2,n_3$ where $n_1=n_2^\gamma$, fix $k_1,k_2$ and $k_3$ such that 
$k_1 k_2 = n_2$, $k_1 = k_2^\gamma$ and $k_3 = n_3$.
Suppose there is an algorithm $\cA$ for $\gouMv$ with parameters $k_1,k_2,k_3$ and 
preprocessing time $p(k_1,k_2)$ and computation time $c(k_1,k_2,k_3)$ that has an error probability of at most $\delta$.
Then
there is an algorithm $\cB$ for $\goMv$ with parameters $n_1,n_2,n_3$ using
preprocessing time $\tilde{O}( n_{1}\cdot p(k_1,k_2 ) )$ and
computation time $\tilde{O}( n_1 \cdot c(k_1,k_2,n_3) )$ that has an error probability of at most $\delta$. 
\end{lemma}

\begin{proof}
Again, we show a reduction for deterministic algorithms. This can be extended to work for randomized algorithms
by a standard application of the Chernoff bound as in the proof of \Cref{oMv<=oMv n*n}. 
By plugging $\cA$ into \Cref{list witness}, we have an algorithm $\cA'$ 
for listing all witnesses of $\gouMv$ with parameters $k_1,k_2,k_3$.
We will formulate $\cB$ using $\cA'$ as a subroutine.
	
Let $M$ be an input matrix $M$ of $\cB$ of size $n_{1}\times n_{2}$, 
we partition $M$ into blocks $\{M_{x,y}\}_{1\le x\le n_{1}/k_1,1\le y\le n_2/k_2}$
each of size $k_1 \times k_2$. 
For each $x$ and $y$, let $I_{x,y}$ be an instance of $\cA'$ and 
we feed $M_{x,y}$ into $I_{x,y}$ to preprocess.
The total preprocessing time is $ \frac{n_1 n_2}{k_1 k_2}\cdot p(k_1,k_2) = n_1 \cdot p(k_1,k_2)$.

In the $\goMv$ problem, for any $t\le n_3$, once $v^{t}$ arrives, we need to compute $b^{t}=Mv^{t}$
before $v^{t+1}$ arrives. To do so, we first partition $v^{t}$ into
blocks $\{v_{y}^{t}\}_{1\le y\le n_2/k_2}$ each of size $k_2$.
We write $b^{t}=\{b_{x}^{t}\}_{1\le x\le n_{1}/k_1}$ as blocks
each of size $k_1$. Note that $b_{x}^{t}=\bigvee_{y}M_{x,y}v_{y}^{t}$
($\bigvee$ means bit-wise OR).

To compute $b_{x}^{t}$, the procedure iterates over all values for $ y $ from $ 1 $ to $ n_2/k_2 $.
When $y=1$, we set $u_{x,1}^{t}$ to be the all-ones vector. 
Let $W_{x,y,t}$ be the set of witnesses $(u_{x}^{t},v_{y}^{t})$.
We feed $(u_{x,y}^{t},v_{y}^{t})$ to the instance $I_{x,y}$ for listing the
witnesses $W_{x,y,t}$.
For all $i\in W_{x,y,t}$, we now know that $(b_{x}^{t})_{i}=1$.
To find other indices $i$ such that $(b_{x}^{t})_{i}=1$, we set
$u_{x,y+1}^{t}$ to be same as $u_{x,y}^{t}$ except that $(u_{x,y+1}^{t})_{i}=0$
for all found witnesses $i\in W_{x,y,t}$. Then we proceed with $y\gets y+1$. We repeat this until $y=n_2/k_2$. Then $b_{x}^{t}$ is completely computed. 
Once this procedure is done for all $x$, $b^{t}$ is completely
computed. We repeat until $t=n_{3}$, and we are done.

Now, we denote $w_{x,y}=\sum_{t}|W_{x,y,t}|$.
By \Cref{list witness} the computation time of the instance $I_{x,y}$ is $ \frac{n_{3}+\tilde{O}(w_{xy})}{n_3} c(k_1,k_2,n_3) $.
Summing over all $x,y$, we have a total running time of
\begin{equation*}
	\sum_{1\le x\le\frac{n_{1}}{k_{1}},1\le y\le\frac{n_{2}}{k_{2}}}\frac{n_{3}+\tilde{O}(w_{xy})}{n_{3}}c(k_{1},k_{2},n_{3}) \, .
\end{equation*}
To conclude that the computation time is $\tilde{O}( n_1 \cdot c(k_1,k_2,n_3) )$, 
it is enough to show that $ \sum_{x, y} w_{x,y} \leq n_1 n_3 $.
Note that for a fixed $ t $, the witness sets $ W_{x,y,t} $ are disjoint for different $x$ and $y$ as we set the entries in the `$u$'-vectors to $0$ for witnesses that we already found.
Furthermore for every $ i \in \bigcup_y W_{x,y,t} $, we have $(b_{x}^{t})_{i}=1$.
As the number of 1-entries of $ b^{t} $ is at most $ n_1 $, we have $ \sum_{x, y} |W_{x,y,t}| \leq n_1 $. 
So $ \sum_{x, y, t} |W_{x,y,t}| \leq n_1 n_3 $. Hence $ \sum_{x, y} w_{x,y} \leq n_1 n_3 $.
\end{proof}

Now we are ready to prove the main theorem.

\begin{proof}
[Proof of \Cref{ouMv hard}]
We will construct an algorithm $\cB$ for $\goMv$ with parameters $n_1,n_2,n_3$ that contradicts \Cref{oMv hard}
by using an algorithm $\cA$ for $\gouMv$ from the statement of \Cref{ouMv hard} as a subroutine.
That is, $\cA$ is an algorithm for $\gouMv$ with parameters $k_1,k_2,k_3$ 
using preprocessing time $p(k_1,k_2)= \poly(k_1,k_1)$ and computation time $c(k_1,k_2,k_3)=\oo(k_1 k_2 k_3)$ that has an error probability of at most $1/3$. We choose $k_1,k_2$ such that $k_1 k_2 = n_2$, $k_1 = k_2^\gamma$ and $k_3 = n_3$.

By \Cref{oMv <= ouMv}, ignoring polylogarithmic factors, $\cB$ has error probability $1/3$ and uses
preprocessing time $\poly(n_{1},n_{2})$
and computation time $n_1 \cdot c(k_1,k_2,n_3) = n_1 \oo(k_1 k_2 n_3) = \oo(n_1 n_2 n_3)$ which contradicts \Cref{oMv hard} by
\Cref{general oMv hard}.
\end{proof}

\subsubsection{Interpreting $\ouMv$ as Graph Problems and a satisfiability problem}
In this section, we show that $\ouMv$ can be viewed as graph problems and satisfiability problem, namely {\em edge query}, {\em independent set query} and {\em 2-CNF query}, as stated in \Cref{equiv ouMv}.
This section is independent from the rest and will not be used later. The edge query problem, however, can be helpful when one wants to show a reduction from $\ouMv$ to graph problems. For example, it is implicit in the reduction from $\ouMv$ to the subgraph connectivity problem. Moreover, the independent set query problem was shown to have an $O(n^2n'/\log^2 n)$ time algorithm by Williams \cite{Williams07} using his \oMv algorithm. The 2-CNF query problem was shown to be equivalent to the independent set query problem in \cite{BansalW12}.
So it is interesting that these problem are equivalent to \ouMv and thus cannot be solved much faster (i.e., in $\oo(n^2n')$ time) assuming the \oMv conjecture. 

\begin{theorem} \label{equiv ouMv}
	For any integers $n$ and $n'$,
	consider the following problems.
	\begin{itemize}
		\item $\ouMv$ with parameters $n_{1},n_{2}=\Theta(n)$ and $n_3 = n'$, preprocessing time $p(n_1,n_2)$ and computation time $c(n_1,n_2,n_3)$.
		\item Independent set (respectively vertex cover) query \cite{BansalW12}: preprocess a graph $G=(V,E)$, when $|V|=n$,
		in time $p_{\mathit{ind}}(n)$. Then given a sequence of sets $S^{1},\dots,S^{n'}\subseteq V$,
		decide if $S^{t}$ is an independent set (respectively a vertex cover) before $S^{t+1}$
		arrives in total time $c_{\mathit{ind}}(n,n')$.
		\item 2-CNF query \cite{BansalW12}: preprocess a 2-CNF $F$ on $n$ variables
		in time $p_{\mathit{cnf}}(n)$. Then given a sequence of assignments $X^{1},\dots,X^{n'}$,
		decide if $F(X^{t})=1$ before $X^{t+1}$	arrives in total time $c_{\mathit{cnf}}(n,n')$.		
		\item Edge query: preprocess a graph $G=(V,E)$, when $|V|=n$, in time $p_{\mathit{edge}}(n)$. Then given
		a sequence of set pairs $(S^{1},T^{1}),\dots,(S^{n'},T^{n'})$,
		decide if there is an edge $(a,b)\in E(S^{t},T^{t})$, before $S^{t+1}$
		arrives in total time $c_{\mathit{edge}}(n,n')$.
	\end{itemize}
	We have that $p(n_1,n_2) = \Theta(p_{\mathit{ind}}(n)) = \Theta(p_{\mathit{cnf}}(n)) = \Theta(p_{\mathit{edge}}(n))$ 
	and $c(n_1,n_2,n_3) = \Theta(c_{\mathit{ind}}(n,n')) = \Theta(c_{\mathit{cnf}}(n,n')) = \Theta(c_{\mathit{edge}}(n,n'))$.
\end{theorem}

Another observation that might be useful in proving \oMv hardness results is the following. 

\begin{theorem}\label{thm:uMv special}
	In the $\ouMv$ problem, we can assume that the matrix $M$ is symmetric, and each vector pair $(u^t,v^t)$ is such that either $u^t=v^t$ or the supports of $u^t$ and $v^t$ are disjoint (i.e., the inner product between $u^t$ and $v^t$ is $0$).
\end{theorem}

The rest of this section is devoted to proving the above theorems. 
First, we need this fact.
\begin{proposition}
\label{lem:special case}Consider a Boolean matrix $M\in\{0,1\}^{n_{1}\times n_{2}}$ and Boolean vectors
$u\in\{0,1\}^{n_{1}}$ and $v\in\{0,1\}^{n_{2}}$. Let $M'=\left[\begin{smallmatrix}
0 & M\\
M^{T} & 0
\end{smallmatrix}\right]$, $w=\left[\begin{smallmatrix}
u\\
v
\end{smallmatrix}\right]$, $x=\left[\begin{smallmatrix}
u\\
0
\end{smallmatrix}\right]$ and $y=\left[\begin{smallmatrix}
0\\
v
\end{smallmatrix}\right]$ where $w,x,y\in\{0,1\}^{n_{1}+n_{2}}$. Then $u^{\top}Mv=w^{\top}M'w=x^{\top}M'y$.\end{proposition}
\begin{proof}
It is easy to verify that $w^{\top}M'w=w^{\top}\left[\begin{smallmatrix}
Mv\\
M^{T}u
\end{smallmatrix}\right]=(u^{\top}Mv)\vee (v^{\top}M^{\top}u)=u^{\top}Mv$. Similarly, $x^{\top}M'y=x^{\top}\left[\begin{smallmatrix}
Mv\\
0
\end{smallmatrix}\right]=u^{\top}Mv$.\end{proof}

\Cref{thm:uMv special} follows immediately from \Cref{lem:special case}. Now we prove \Cref{equiv ouMv}.

\begin{proof}[Proof of \Cref{equiv ouMv}.]
Let $\cA_{\oMv}, \cA_{\mathit{ind}}$ and $\cA_{\mathit{edge}}$
be the algorithms for $\oMv$, independent set query and edge query respectively.

\textit{($\ouMv \Rightarrow$ independent set query)} Given
an input graph $G=(V,E)$ of $\cA_{\mathit{ind}}$, preprocess the adjacency matrix $M$ of $G$ using $\cA_{\oMv}$. 
Once $S^{t}$ arrives, let $v^{t}$ be the indicator vector of $S^{t}$ 
(i.e., for all $i\in V$, $(v^t)_i=1$ if $i\in S^t$, otherwise $(v^t)_i=0$).
Observe that $(v^{t})^{\top}Mv^{t}=0$ iff $S^{t}$ is independent.
So we can use $\cA_{\oMv}$ to answer the query of $\cA_{\mathit{ind}}$.

\textit{($\ouMv \Rightarrow$ edge query)} Given an input graph $G$ of $\cA_{\mathit{edge}}$, 
preprocess the adjacency matrix $M$ of $G$ using  $\cA_{\oMv}$.
Once $(S^{t},T^{t})$ arrives, let $u^{t}$ and $v^{t}$ be the indicator vectors of $S^{t}$ and $T^{t}$ respectively. 
Observe that $(u^{t})^{\top}Mv^{t}=1$ iff there is an edge $(a,b)\in E(S^{t},T^{t})$.
So we can use $\cA_{\oMv}$ to answer the query of $\cA_{\mathit{edge}}$.

\textit{(independent set query $\Rightarrow \ouMv$)} Given an input matrix $M$ of $\cA_{\oMv}$, 
let $G$ the graph defined by the adjacency matrix
$M'=\left[\begin{smallmatrix}
0 & M\\
M^{T} & 0
\end{smallmatrix}\right]$. We preprocess $G$ using  $\cA_{\mathit{ind}}$.
Once $(u^{t},v^{t})$ arrives, let $S^{t}$ be the set indicated by
$w=\left[\begin{smallmatrix}
u^{t}\\
v^{t}
\end{smallmatrix}\right]$
(i.e., $i\in S^t$ iff $w_i$=1). 
We have that $S^{t}$ is independent iff $(u^{t}){}^{\top}Mv^{t}=w^{\top}M'w=0$
by  \Cref{lem:special case}.
So we can use $\cA_{\mathit{ind}}$ to answer the query of $\cA_{\oMv}$.

\textit{(edge query $\Rightarrow \ouMv$)} Given an input matrix $M$ of $\cA_{\oMv}$, 
let $G$ be the graph defined by the adjacency matrix $M'=\left[\begin{smallmatrix}
0 & M\\
M^{T} & 0
\end{smallmatrix}\right]$. We preprocess $G$ using  $\cA_{\mathit{edge}}$.
Once $(u^{t},v^{t})$ arrives, let $S^{t},T^{t}$ be the sets indicated
by $x=\left[\begin{smallmatrix}
u^{t}\\
0
\end{smallmatrix}\right]$ and $y=\left[\begin{smallmatrix}
0\\
v^{t}
\end{smallmatrix}\right]$. There is an edge $(a,b)\in E(S^{t},T^{t})$ iff $(u^{t}){}^{\top}Mv^{t}=x^{\top}M'y=1$
by  \Cref{lem:special case}. 
So we can use $\cA_{\mathit{edge}}$ to answer the query of $\cA_{\oMv}$.

(independent set query $\Leftrightarrow$ 2-CNF query) See \cite[Section 2.3]{BansalW12}.
\end{proof}

Note that we can use an $\oMv$ algorithm to solve the dominating set query
problem, defined in a similar way as independent set query problem. 
Indeed, let $M$ be the adjacency matrix of $G$, and $v$ be an indicator vector of
$S$. We have that $Mv\vee v$ (bit-wise OR) is the all-one vector iff $S$ is a dominating
set. However, it is not clear if the reverse reduction exists.

\section{Hardness for Amortized Fully Dynamic and Worst-case Partially
Dynamic Problems}\label{sec:full:fully hardness}

In this section, we give reductions from our intermediate problems to various dynamic problems.
In \Cref{sec:lb_high_query}, we give conditional lower bounds for those graph problems whose algorithms cannot have the update time $u(m)=\oo(\sqrt{m})$ and the query time $q(m)=\oo(m)$ simultaneously. 
In \Cref{sec:lb_trade_off}, we give the bounds for those problems that cannot have the update time $u(m)=\oo(m^{1-\delta})$ and the query time $q(m)=\oo(m^\delta)$ simultaneously, for any constant $0<\delta<1$.
In \Cref{sec:lb_graph_other,sec:lb_non_graph}, we give the lower bounds for the remaining graph and non-graph problems, whose lower bound parameters of update/query time are in a different form (see \Cref{fig:lb_plot}).
We devote \Cref{sec:lb_diam,sec:lb_densest} to proving the lower bounds for approximating the diameter of a weighted graph and the densest subgraph problem, respectively, because their reductions are more involved.

\begin{figure}
	\centering
	\includegraphics[scale=0.5]{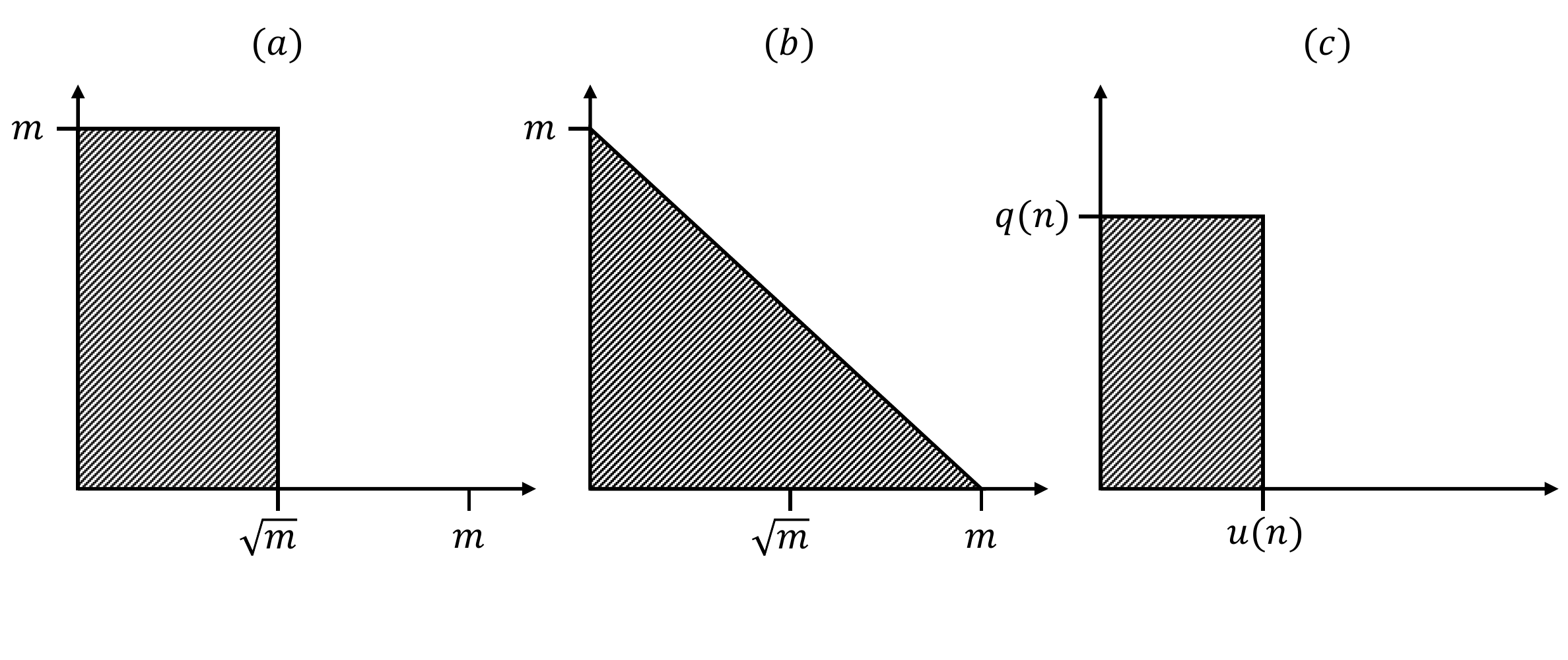} 
	
	\caption{\small
		Horizontal and vertical axes represent the update and query time of dynamic algorithms respectively. The shaded areas indicate the ranges of update/query time whose existence of dynamic algorithm with such parameters would contradict \Cref{oMv hard}. Chart (a) and (b) are for problems in \Cref{sec:lb_high_query} and \Cref{sec:lb_trade_off} respectively. Chart (c) is for problems in \Cref{sec:lb_graph_other,sec:lb_non_graph,sec:lb_diam,sec:lb_densest} where the lower bound parameters are in many different forms.
	}\label{fig:lb_plot}
\end{figure}

Our hardness results, compared to previously known bounds, for fully dynamic problems are summarized in \Cref{table: summary fully graph} and \Cref{table: summary fully non-graph}. 
Our tight hardness results are summarized in \Cref{table:tight results}.

Given a matrix $M\in\{0,1\}^{n_{1}\times n_{2}}$, we denote by $G_{M}=((L,R),E)$ the bipartite
graph where $L=\{l_{1},\dots,l_{n_{1}}\}$, $R=\{r_{1},\dots,r_{n_{2}}\}$,
and $E=\{(r_{j},l_{i})\mid M_{ij}=1\}$.

\begin{table}
\footnotesize
\begin{tabular}{|>{\centering}p{0.25\textwidth}|>{\centering}p{0.07\textwidth}|>{\centering}p{0.07\textwidth}|>{\centering}p{0.07\textwidth}|>{\centering}p{0.07\textwidth}|>{\centering}p{0.07\textwidth}|>{\centering}p{0.07\textwidth}|>{\centering}p{0.15\textwidth}|}
\hline 
\multirow{2}{0.25\textwidth}{Problems} & \multicolumn{3}{c|}{Upper Bounds} & \multicolumn{3}{c|}{Lower Bounds} & \multirow{2}{0.15\textwidth}{Remark}\tabularnewline
\cline{2-7} 
 & $p(m,n)$ & $u(m,n)$ & $q(m,n)$ & $p(m,n)$ & $u(m,n)$ & $q(m,n)$ & \tabularnewline
\hline 
\hline 
\multirow{4}{0.25\textwidth}{ss-SubConn, ap-SubConn} & $m^{4/3}$ & $m^{2/3}$ & $m^{1/3}$ & $poly$ & $m^{2/3-\eps}$ & $m^{1/3-\eps}$ & Upper: \cite{ChanPR11}, amortized only; Lower: when $m=O(n^{3/2})$ \tabularnewline
\cline{2-8} 
& $m^{6/5}$ & $m^{4/5}$ & $m^{1/5}$ & $poly$ & $m^{4/5-\eps}$ & $m^{1/5-\eps}$ & Upper: \cite{Duan10}; Lower: when $m= O(n^{5/4})$ \tabularnewline
\cline{2-8} 
 & $1$ & $1$ & $m$ & $poly$ & $m^{1/2-\eps}$ & $m^{1-\eps}$ & Lower: when $m= O(n^2)$ \tabularnewline
\cline{2-8} 
 & $m$ & $n$ & 1  & $poly$ & $n^{1-\eps}$ & $n^{2-\eps}$ & Upper: \cite{KapronKM13}; Lower: $m=\Theta(n^{2})$\tabularnewline
\hline 
st-SubConn, (unweighted) st-SP, st-Reach, s-triangle detection, strong connectivity & 1 & 1 & $m$ & $poly$ & ${m}^{1/2-\eps}$ & $m^{1-\eps}$ & Lower: when $m= O(n^2)$\tabularnewline
\hline 
\multirow{2}{0.25\textwidth}{(unweighted) ss-SP, ss-Reach} & 1 & $m$ & 1 & $poly$ & $m^{1-\epsilon}$ & $m^{\epsilon'}$ & Lower: when $m= O(n^{1/(1-\eps)})$, $\epsilon'< \epsilon$ \tabularnewline
\cline{2-8} 
 & 1 & 1 & $m$ & $poly$ & $m^{1/2-\eps}$ & $m^{1-\eps}$ & Lower: when $m= O(n^2)$ \tabularnewline
\hline 
Erickson's problem & 1 & $n$ & 1 & $poly$ & $n^{1-\eps}$ & $n^{1-\eps}$ & Upper: binary search tree \tabularnewline
\hline 
$d$-failure connectivity  & $mn^{1/c}$ & $d^{2c+4}$ & $d$ & $poly(n)$ & $poly(d)$ & $d^{1-\eps}$ & Upper: \cite{DuanP10} \tabularnewline
\hline 
$3$-approx vertex color distance oracle & $m\sqrt{n}$ & $\sqrt{n}$ & $1$ & $poly$ & $n^{1-\eps}$ & $n^{2-\eps}$ & Upper: \cite{HermelinLWY11,Chechik12,LackiOPSZ13}; Lower: when $(3-\epsilon)$-approx \tabularnewline
\hline 
Pagh's problem over $k$ sets in a universe $[n]$ & $1$ & $n$ & $1$ & $poly$ & $n^{1-\eps}$ & $k^{1-\eps}$ & \tabularnewline
\hline 

Multiphase over $k$ sets in a universe $[n]$ & $1$ & \multicolumn{2}{c|}{$\tau \le \max\{k,n\}$} & $poly$ & \multicolumn{2}{c|}{$\tau \ge \min\{k,n\}^{1-\eps}$}  & \tabularnewline
\hline 
\end{tabular}
\caption{Our tight results along with the matching upper bounds (or better upper bounds when worse approximation ratio is allowed).
The polylogarithmic factors are omitted. 
The lower bounds state that there is no algorithm achieving stated preprocessing time, amortized update time, and query time simultaneously, unless the \oMv conjecture fails.
The matching upper bounds of update time are all worst-case time except the bound by \cite{ChanPR11}. The upper bounds without remark are the naive ones.
}

\label{table:tight results}
\end{table}

\begin{table}
\footnotesize
\begin{tabular}{|>{\centering}p{0.3\textwidth}|>{\centering}p{0.08\textwidth}|>{\centering}p{0.08\textwidth}|>{\centering}p{0.08\textwidth}|>{\centering}p{0.08\textwidth}|>{\centering}p{0.12\textwidth}|>{\centering}p{0.20\textwidth}|}
\hline 
Problems & $p(m,n)$ & $u(m,n)$ & $q(m,n)$ & Conj.& Reference & Remark\tabularnewline
\hline 
\hline 
\multirow{3}{0.3\textwidth}{st-SubConn} & $m^{4/3}$ & $m^{\delta-\epsilon}$ & $m^{2/3-\delta-\epsilon}$ & 3SUM & \cite{AbboudW14} &Choose any $\delta\in[1/6,1/3]$. $m \le O(n^{1.5})$ \tabularnewline
\cline{2-7} 
 & $m^{1+\delta-\epsilon}$ & $m^{\delta-\epsilon}$ & $m^{2\delta-\epsilon}$ & Triangle & \cite{AbboudW14} & For some $\delta<0.41$ depending on the conjecture\tabularnewline
\cline{2-7} 
 & $n^{3-\epsilon}$ & $n^{1-\epsilon}$ & $n^{2-\epsilon}$ & BMM & \cite{AbboudW14} & When $m=\Theta(n^{2})$\tabularnewline
\hline 
\multirow{3}{0.3\textwidth}{st-Reach} & $m^{4/3}$ & $m^{\delta-\epsilon}$ & $m^{2/3-\delta-\epsilon}$ & 3SUM & \cite{AbboudW14} & Choose any $\delta\in[1/6,1/3]$. $m \le O(n^{1.5})$ \tabularnewline
\cline{2-7} 
 & $\mathbf{m^{1+\delta-\epsilon}}$ & $\mathbf{m^{2\delta-\epsilon}}$ & $\mathbf{m^{2\delta-\epsilon}}$ (*) & Triangle & \cite{AbboudW14} & Only lower bound for amortized time over $O(n)$ updates; for some $\delta<0.41$
depending on the conjecture \tabularnewline
\cline{2-7} 
 & $\mathbf{n^{3-\epsilon}}$ & $\mathbf{n^{2-\epsilon}}$ & $\mathbf{n^{2-\epsilon}}$ (*) & BMM & \cite{AbboudW14} & Only lower bound for amortized time over $O(n)$ updates; $m=\Theta(n^{2})$ \tabularnewline
\hline 
st-SubConn\tablefootnote{implies the same bound for all reachablity problems (including transitive closure), strong connectivity, bipartite perfect matching, size of maximum matching, minimum vertex cover, maximum independent set on bipartite graph, size of st-maxflow on undirected unit capacity. See some reductions from \cite{AbboudW14}}, st-Reach, unweighted $(\alpha,\beta)$st-SP where $3\alpha+\beta<5$
\tablefootnote{implies all shortest path problems (ss-SP,ap-SP) with same approximation factor},
Triangle Detection, s-Triangle Detection\tablefootnote{implies the same bound for st-Reach}, etc. (See footnotes) & $\mathbf{\poly(n)}$ & $\mathbf{m^{1/2-\epsilon}}$ & $\mathbf{m^{1-\epsilon}}$ & $\oMv$ & \Cref{cor:high query} & Choose any $m\le n^{2}$\tabularnewline
\hline 
ss-SubConn\tablefootnote{implies the same bound for ap-SubConn, ss-Reach, transitive closure.},
unweighted $(\alpha,\beta)$ss-SP where $2\alpha,\beta<4$\tablefootnote{implies the same bound for ap-SP with same approximation factor.},
unweighted $(\alpha,\beta)$vertex color distance oracle where $\alpha+\beta<3$, etc. (See footnotes) & $\mathbf{\poly(n)}$ & $\mathbf{m^{\delta-\epsilon}}$ & $\mathbf{m^{1-\delta-\epsilon}}$ & $\oMv$ & \Cref{corr:trade-off}  &Choose any $\delta\in(0,1)$, and $m\le\min\{n^{1/\delta},n^{1/(1-\delta)}\}$\tabularnewline
\hline 
unweighted $(3-\epsilon)$st-SP, $(2-\epsilon)$diameter on weighted Graphs,
Densest Subgraph of size at least 5 & $\mathbf{\poly(n)}$ & $\mathbf{n}^{\mathbf{1/2-\epsilon}}$ & $\mathbf{n^{1-\epsilon}}$ & $\oMv$ &
Corollaries \ref{corr:3approx stSP}, \ref{corr:approx diam} and \ref{corr:densest} & \tabularnewline
\hline 
\multirow{3}{0.3\textwidth}{$d$-failure Connectivity} & $n^{2-\epsilon}$ & $(dn)^{1/2-\epsilon}$ & $d^{1/2-\epsilon}$ & 3SUM  & \cite{KopelowitzPP14} & \tabularnewline
\cline{2-7} 
 & $\mathbf{\poly(n)}$ & $\mathbf{(dn)^{1-\epsilon}}$ & $\mathbf{d^{1-\epsilon}}$ & $\oMv$ & \Cref{corr:d failure dn}& Choose any $\delta\in(0,1/2]$, $d=m^\delta$, $m=\Theta(n^{1/(1-\delta)})$ \tabularnewline
\cline{2-7} 
 & $\mathbf{\poly(n)}$ & $\mathbf{d^{1/\delta}}$ & $\mathbf{d^{1-\epsilon}}$ & $\oMv$ & \Cref{corr:d failure d poly} & Choose any $\delta\in(0,1/2]$, $d=m^\delta$, $m\le \Theta(n^{1/(1-\delta)})$ \tabularnewline
\hline 
\end{tabular}

\caption{Amortized lower bounds for fully dynamic graph problems and worst-case lower bounds for partially dynamic graph problems.
Bounds which are not subsumed are highlighted.
Each row states that there is no algorithm achieving stated preprocessing
time, update time, and query time \emph{simultaneously}, unless the
conjecture fails. 
Except for $(2-\epsilon)$-approx diameter on weighted graphs,
all the lower bounds also hold for the worst-case update time of partially dynamic algorithms.
Bounds marked with the asterisk (*) hold when the update time is amortized over \emph{only $O(n)$ updates}. It is not clear how to get this parameter for an update time amortized over any polynomially many updates like all our bounds.
}
\label{table: summary fully graph}
\end{table}

\begin{table}
\footnotesize
\begin{tabular}{|>{\centering}p{0.3\textwidth}|>{\centering}p{0.08\textwidth}|>{\centering}p{0.08\textwidth}|>{\centering}p{0.08\textwidth}|>{\centering}p{0.08\textwidth}|>{\centering}p{0.12\textwidth}|>{\centering}p{0.20\textwidth}|}
	\hline 
	Problems & $p(m,n)$ & $u(m,n)$ & $q(m,n)$ & Conj.& Reference & Remark\tabularnewline
\hline 
\hline 
\multirow{2}{0.3\textwidth}{Langerman's } & $n^{1+\delta-\epsilon}$ & $n^{\delta/2-\epsilon}$ & $n^{\delta/2-\epsilon}$ & multi phase &
\cite{Patrascu10} &  If $\tau\ge n^{\delta-\epsilon}$, $\delta\in(0,1)$ \tabularnewline
\cline{2-7} 
 & $\mathbf{\poly(n)}$ & $\mathbf{n^{1/2-\epsilon}}$ & $\mathbf{n^{1/2-\epsilon}}$ & $\oMv$ & \Cref{corr:langerman} &\tabularnewline
\hline 
\multirow{4}{0.3\textwidth}{Pagh's over $k$ sets in a universe $[n]$} & $k^{1+\delta-\epsilon}$ & $k^{\delta-\epsilon}$ & $k^{\delta-\epsilon}$ & Triangle & \cite{AbboudW14} & $n\le k\le n^{2}$, Choose $\delta\in(1/3,0.41)$ depending on the
conjecture\tabularnewline
\cline{2-7} 
 & $k^{3/2-\epsilon}$ & $k^{1/2-\epsilon}$ & $k^{1/2-\epsilon}$ & BMM & \cite{AbboudW14} &$n\le k\le n^{2}$\tabularnewline
\cline{2-7} 
 & $k^{1/3-\epsilon}$ & $k^{1/3-\epsilon}$ & $k^{1/3-\epsilon}$ & 3SUM & \cite{AbboudW14} & $k=\Theta(n^{1.5+\epsilon'})$\tabularnewline
\cline{2-7} 
 & $\mathbf{\poly(n)}$ & $\mathbf{n^{1-\epsilon}}$ & $\mathbf{k^{1-\epsilon}}$ & $\oMv$ & \Cref{corr:pagh} &$k\le poly(n)$ \tabularnewline
\hline 
\multirow{2}{0.3\textwidth}{Erickson over a matrix of size $n\times n$} & $n^{2+\delta-\epsilon}$ & $n^{\delta-\epsilon}$ & $n^{\delta-\epsilon}$ & multi phase & \cite{Patrascu10} &If $\tau\ge n^{\delta-\epsilon}$, $\delta\in(0,1)$ \tabularnewline
\cline{2-7} 
 & $\mathbf{\poly(n)}$ & $\mathbf{n^{1-\epsilon}}$ & $\mathbf{n^{2-\epsilon}}$ & $\oMv$ & \Cref{corr:erickson} &\tabularnewline
\hline 
\multirow{2}{0.3\textwidth}{Multiphase over $k$ sets in a universe $[n]$} & $n^{4-\epsilon}$ & \multicolumn{2}{c|}{$\tau\ge n^{1/2-\epsilon}$} & 3SUM & \cite{Patrascu10} & $k=\Theta(n^{5/2})$\tabularnewline
\cline{2-7} 
 & $\mathbf{\poly(k,n)}$ & \multicolumn{2}{c|}{$\mathbf{\tau\ge}\mathbf{min\{k,n\}^{1-\epsilon}}$} & $\oMv$ & \Cref{corr:multiphase}&\tabularnewline
\hline 
\end{tabular}

\caption{Amortized lower bounds for fully dynamic non-graph problems. 
Bounds which are not subsumed are highlighted. 
Each row states that there is no algorithm achieving stated preprocessing
time, update time, and query time \emph{simultaneously}, unless the
corresponding conjecture fails.
The lower bounds based on the multiphase problem are
only for worst case time and their parameters are implicit from \cite{Patrascu10}.}

\label{table: summary fully non-graph}
\end{table}

In this section, our proofs usually follow two simple steps.
First, we show the reductions in lemmas that 
given a dynamic algorithm $\cA$ for some problem, one can solve $\uMv$ by running the preprocessing step of $\cA$ on some graph 
and then making
some number of updates and queries.
Then, we conclude in corollaries that if either 1)
$\cA$ has low worst-case update/query time, or 2) $\cA$ has low amortized update/query time and $\cA$ is fully dynamic,
then this contradicts \Cref{oMv hard}.

\subsection{Lower Bounds for Graph Problems with High Query Time}\label{sec:lb_high_query}
To show hardness of the problems in this section, we reduce from $\uMv$ where $n_1 = n_2 = \sqrt{m}$. 
The idea is that when $u$ and $v$ arrive, we make the update operations of the dynamic algorithm $\cA$
to ``handle'' both $u$ and $v$. Then make only 1 query to $\cA$ to answer $1$-$\uMv$.
Since the reduction is efficient in the number of queries, we get a high lower bound of query time.

\paragraph{$s$-$t$ Subgraph Connectivity (st-SubConn)}
\begin{lemma}
\label{st-subconn reduc}Given a partially dynamic algorithm $\cA$ for
st-SubConn, one can solve $1$-$\uMv$ with parameters $n_1$ and $n_2$ by 
running the preprocessing step of $\cA$ on a graph with $O(m)$ edges and $\Theta(\sqrt{m})$ vertices, 
and then performing $O(\sqrt{m})$ turn-on operations (or $O(\sqrt{m})$ turn-off operations) and $1$ query,
where $m$ is such that $n_1=n_2=\sqrt{m}$.
\end{lemma}
\begin{proof}
We only prove the decremental case, because the incremental case is
symmetric. 
Given $M$, we construct the bipartite graph $G_{M}$ and add to it
vertices $s,t$, and edges $(t,l_{i}),(r_{j},s)$ for all $r_{j}\in R,l_{i}\in L$. Thus, the total number of edges is at most $n_1 n_2 + n_1 + n_2 = O(m)$. 
In the beginning, every vertex is ``turned on'', i.e., included in the set $ S $ of the st-SubConn algorithm

Once $u$ and $v$ arrive, we turn off $l_{i}$ iff $u_{i}=0$ and turn
off $r_{j}$ iff $v_{j}=0$. We have $u^{\top}Mv=1$ iff $s$ is connected
to $t$. In total, we need to do at most $n_1+n_2 = O(\sqrt{m})$ updates
and $1$ query.
\end{proof}

\paragraph{Distinguishing between 3 and 5 for $s$-$t$ distance (st-SP (3 vs.\ 5))}
\begin{lemma}
\label{st-SP 3vs5 reduc}Given a partially dynamic algorithm $\cA$ for
$(\alpha,\beta)$-approximate st-SP with $3\alpha+\beta<5$, one can
solve $1$-$\uMv$ with parameters $n_1$ and $n_2$ by 
running the preprocessing step of $\cA$ on a graph with $O(m)$ edges and $\Theta(\sqrt{m})$ vertices, 
and then making $O(\sqrt{m})$ insertions (or $O(\sqrt{m})$ deletions) and $1$ query,
where $m$ is such that $n_1=n_2=\sqrt{m}$. \end{lemma}
\begin{proof}
We only prove the decremental case, because the incremental case is
symmetric. 
Given $M$, we construct the bipartite graph $G_{M}$ and add to it
vertices $s,t$, and edges $(t,l_{i}),(r_{j},s)$ for all $r_{j}\in R,l_{i}\in L$. Thus, the total number of edges is at most $n_1n_2+n_1+n_2 = O(m)$. 

Once $u$ and $v$ arrive, we delete $(t,l_{i})$ iff $u_{i}=0$ and
delete $(r_{j},s)$ iff $v_{j}=0$. If $u^{\top}Mv=1$, then $d(s,t)=3$,
otherwise $d(s,t)\ge5$. In total, we need to do at most $n_1+n_2=O(\sqrt{m})$
updates and $1$ query.
\end{proof}

\paragraph{Triangle Detection and Triangle Detection at vertex $s$}
\begin{lemma}
\label{tri reduc}Given a partially dynamic algorithm $\cA$ for
(s-)triangle detection, one can solve $1$-$\uMv$ with parameters $n_1$ and $n_2$ by
running the preprocessing step of $\cA$ on a graph with $O(m)$ edges and $\Theta(\sqrt{m})$ vertices, 
and then making $O(\sqrt{m})$ insertions (or $O(\sqrt{m})$ deletions) and $1$ query,
where $m$ is such that $n_1=n_2=\sqrt{m}$. \end{lemma}
\begin{proof}
We only prove the decremental case. Given $M$, we construct the bipartite
graph $G_{M}$ and add to it a vertex $s$ and edges $(s,l_{i}),(r_{j},s)$
for all $r_{j}\in R,l_{i}\in L$. Thus, the total number of edges is at most $n_1n_2+n_1+n_2 = O(m)$. 

Once $u$ and $v$ arrive, we delete $(s,l_{i})$ iff $u_{i}=0$ and
delete $(r_{j},s)$ iff $v_{j}=0$. We have $u^{\top}Mv=1$ iff there
is a triangle in a graph iff there is a triangle incident to $s$.
In total, we need to do $n_1+n_2=O(\sqrt{m})$ updates and $1$
query.\end{proof}
\begin{corr}
	\label{cor:high query}
	For any $n$ and $m \le n^2$,
	unless \Cref{oMv hard} fails, there is no partially dynamic algorithm $\cA$
	for the problems in the list below for graphs with $n$ vertices and $m$ edges
	with preprocessing time $p(m)=poly(m)$,
	worst update time $u(m)=\oo(\sqrt{m})$ and query time $q(m)=\oo(m)$
	that has an error probability of at most $1/3$.
	Moreover, this is true also for fully dynamic algorithms with amortized update time.
	The problems are:
	\begin{itemize}[noitemsep,nolistsep]
	\item st-SubConn
	\item st-SP (3 vs.\ 5)
	\item (s-)triangle detection
	\end{itemize}
\end{corr}
\begin{proof}
	Suppose there is such a partially dynamic algorithm $\cA$. That is, on a graph with $n_0$ vertices and $m_0 =  O(n_0^2)$ edges, $\cA$ has worst-case update time 
	$u(m_0)=\oo(\sqrt{m_0})$ and query time $q(m_0)=\oo(m_0)$.
	We will construct an algorithm $\cB$ for $1$-$\uMv$ with parameters $n_1$ and $n_2$ which contradicts \Cref{oMv hard}. 
	Using \Cref{st-subconn reduc,st-SP 3vs5 reduc,tri reduc}, by running $\cA$ on a graph with $n_0 = \Theta(\sqrt{m})$ vertices and $m_0 = O(m)$ edges where $m$ is such that $n_1=n_2=\sqrt{m}$ 
	(note that, indeed, $m_0 =  O(n_0^2)$),
	$\cB$ has preprocessing time $\poly(m)$ 
	and computation time $O( \sqrt{m} u(m) +  q(m) ) = \oo(m)$ which contradicts \Cref{oMv hard} by \Cref{uMv hard}.
		
	Next, suppose that $\cA$ is fully dynamic and only guarantees an amortized bound. 
	We will construct an algorithm $\cC$ for $1$-$\ouMv$ with parameters $n_1$, $n_2$, and $n_3$ which again contradicts \Cref{oMv hard}
	by running $\cA$ on the same graph as for solving 1-$\uMv$ while the number of updates and queries needed is multiplied by $O(n_3)$.
	This can be done because $\cA$ is fully dynamic. 
	So, for each vector pair $(u,v)$ for $\cC$, if $\cA$ makes $k$ updates to the graph,
	then $\cA$ can undo these updates with another $k$ updates so that the updated graph is the same as right after the preprocessing.
	Recall that, by the notion of amortization, if there are $t$ updates, 
	then $\cA$ takes $O( (t+m_0) \cdot u(m_0) )$ time where $m_0$ is a number of edges ever appearing in the graph.
	By choosing $n_3 = \sqrt{m}$, we have that $\cC$ has preprocessing time $\poly(m)$ 
	and computation time $O( (\sqrt{m}n_3 + m )u(m) + n_3 q(m) ) = \oo(m\sqrt{m})$ which contradicts \Cref{oMv hard} by \Cref{ouMv hard}.
\end{proof}
Note that st-SubConn is reducible to the following problems in a way that preserves the parameters of the lower bounds (see \cite{AbboudW14} for the first three reductions):
\begin{itemize}[noitemsep,nolistsep]
\item st-Reach,
\item Strong connectivity,
\item Bipartite perfect matching,
\item Size of bipartite maximum matching (and, hence, vertex cover),
\item st-maxflow in undirected and unit capacity graph (see \cite[Theorem 3.6.1]{Madry11}).
\end{itemize}
Therefore, these problems have the same lower bound.

\subsection{Lower Bounds for Graph Problems with a Trade-off}\label{sec:lb_trade_off}
To show hardness of the problems in this section, we reduce from $\uMv$ where $n_1 = m^\delta$ and $n_2 = m^{1-\delta}$ for any constant $\delta \in (0,1)$.
When $u$ and $v$ arrive, we make $m^{1-\delta}$ updates of the dynamic algorithm $\cA$ to handle only $v$. 
Then make $m^\delta$ queries to $\cA$ to find the value of $\uMv$.
Since the choice of $\delta$ is free, we get a trade-off lower bound between update time and query time.

Through out this subsection, $\delta\in(0,1)$ is any constant.

\paragraph{Single Source Subgraph Connectivity (ss-Subconn) }
\begin{lemma}
	\label{lem:ss-subconn reduc}Given a partially dynamic algorithm $\cA$ for ss-SubConn, after polynomial preprocessing time, 
	one can solve $(\frac{1-\delta}{\delta})$-$\uMv$ with parameters $n_1$ and $n_2$ by running the preprocessing step of $\cA$ on a graph with $\Theta(m^{1-\delta}+m^{\delta})$ nodes and $O(m)$ edges,
	then making $O(m^{1-\delta})$ insertions (or $O(m^{1-\delta})$ deletions) and $O(m^{\delta})$ queries,
	where $m$ is such that $m^\delta = n_1$ (so $m^{1-\delta} = n_2$).
	\end{lemma}
\begin{proof}
	We only prove the decremental case because the incremental case is
	symmetric. Given $M$, we construct the bipartite
	graph $G_{M}$, with an additional vertex $s$ and edges $(r_{j},s)$
	for all $j\le n_2$. Thus, the total number of edges is $n_1 n_2 + n_2= O(m)$.
	In the beginning, every node is turned on.
	Once $u$ and $v$ arrive, we turn off $r_{j}$ iff $v_{j}=0$. If $u^{\top}Mv=1$,
	then $s$ is connected to $l_{i}$ for some $i$ where $u_{i}=1$.
	Otherwise, $s$ is not connected to $l_{i}$ for all $i$ where $u_{i}=1$.
	We distinguish these two cases by querying, for every $ 1 \leq i \leq n_2 $, whether $ s $ and $ l_i $ are connected.
	In total, we need to do $n_2=m^{1-\delta}$ updates and $n_1=m^{\delta}$
	queries.
\end{proof}

\paragraph{Distinguishing between 2 and 4 for distances from $s$ (ss-SP (2 vs.\ 4))}
\begin{lemma}
\label{ss-SP 2vs4 reduc}Given a partially dynamic algorithm $\cA$ for
$(\alpha,\beta)$-approximate ss-SP with $2\alpha+\beta<4$, one can
solve $(\frac{1-\delta}{\delta})$-$\uMv$ with parameters $n_1$ and $n_2$ 
by running the preprocessing step of $\cA$ on a graph with $O(m)$ edges and $O(m^\delta + m^{1-\delta})$ vertices,
and then making $O(m^{1-\delta})$ insertions (or $O(m^{1-\delta})$ deletions) and $O(m^{\delta})$
queries, 
where $m$ is such that $m^\delta = n_1$ (so $m^{1-\delta} = n_2$).
\end{lemma}
\begin{proof}
We only prove the decremental case. Given $M$, we construct the bipartite
graph $G_{M}$ and add to it a vertex $s$ and edges $(r_{j},s)$
for all $r_j \in R$. Thus, the total number of edges $n_1n_2+n_2=O(m)$. 

Once $u$ and $v$ arrive, we disconnect $s$ from $r_{j}$ iff $v_{j}=0$.
We have that if $u^{\top}Mv=1$, then $d(s,l_{i})=2$ for some $i$
where $u_{i}=1$ , otherwise $d(s,l_{i})\ge4$ for all $i$ where
$u_{i}=1$. In total, we need to do $n_2=O(m^{1-\delta})$ updates and
$n_1=O(m^{\delta})$ queries.
\end{proof}

\paragraph{Vertex-color Distance Oracle (1 vs.\ 3)}
Vertex-color distance oracles are studied in \cite{HermelinLWY11,Chechik12}. 
Given a graph $G$, one can change the color of any vertex and must handle the query that, for any vertex $u$ and color $c$, 
return $d(u,c)$ the distance from $u$ to the nearest vertex with color $c$. 
Chechik \cite{Chechik12} showed, for any integer $k\ge2$, a dynamic oracle with update time $\tilde{O}(n^{1/k})$ and query time $O(k)$ which $(4k-5)$ approximates the distance. 
Lacki~et~al.~\cite{LackiOPSZ13} extended the result when $k=2$ by handling additional operations and used it as a subroutine to get an algorithm for dynamic $(6+\epsilon)$-approximate Steiner tree. 

\begin{lemma}
\label{colDistOracle reduc}Given a dynamic algorithm $\cA$ for $(\alpha,\beta)$-approximate
vertex-color distance oracle with $\alpha+\beta<3$, one can solve $(\frac{1-\delta}{\delta})$-$\uMv$ with parameters $n_1$ and $n_2$ 
by running the preprocessing step of $\cA$ on a graph with $O(m)$ edges and $O(m^\delta + m^{1-\delta})$ vertices,
and then making $O(m^{1-\delta})$ vertex-color changes and $O(m^{\delta})$ queries,
where $m$ is such that $m^\delta = n_1$ (so $m^{1-\delta} = n_2$).
 \end{lemma}
\begin{proof}
Given $M$, we construct the bipartite graph $G_{M}$. 
Set the color of $l_{i}$ for all $l_i \in L$ to~$c$.
The colors of the other vertices (i.e., those in $ R $) are set to $c' \neq c$.

Once $u$ and $v$ arrive, we set the color of $l_{i}$ to $c'\neq c$
iff $u_{i}=0$. If $u^{\top}Mv=1$, then $d(r_{j},c)=1$ for some
$j$ where $v_{j}=1$. Otherwise, $d(r_{j},c)\ge3$ for all $j$ where
$v_{j}=1$. In total, we need to do $n_1=O(m^{\delta})$
updates and $n_2= O(m^{1-\delta})$ queries.\end{proof}

\begin{corr}
\label{corr:trade-off}
For any $n$, $m \le O(\min\{n^{1/\delta}, n^{1/(1-\delta)}\})$, and constant $\delta\in(0,1)$, 
unless \Cref{oMv hard} fails, there is no partially dynamic algorithm $\cA$
for the problems in the list below for graphs with $n$ vertices and at most $m$ edges
with preprocessing time $p(m)=poly(m)$, 
worst-case update time $u(m)=\oo(m^{\delta})$, and query time $q(m)=\oo(m^{1-\delta})$
that has an error probability of at most $1/3$.
Moreover, this is true also for fully
dynamic algorithms with amortized update time. 
The problems are:
\begin{itemize}[noitemsep,nolistsep]
\item ss-SubConn
\item ss-SP (2 vs.\ 4)
\item Vertex-color Distance Oracle (1 vs.\ 3)
\end{itemize}
\end{corr}

\begin{proof}	
	Suppose there is such a partially dynamic algorithm $\cA$. That is, on a graph with $n_0$ vertices and $m_0 =  O(\min\{n_0^{1/\delta}, n_0^{1/(1-\delta)}\})$ edges, $\cA$ has worst-case update time 
	$u(m_0)=\oo(m_0^{\delta})$ and query time $q(m_0)=\oo(m_0^{1-\delta})$.
	We will give an algorithm $\cB$ for $(\frac{1-\delta}{\delta})$-$\uMv$ with parameters $n_1$ and $n_2$ which contradicts \Cref{oMv hard}.
	Using \Cref{lem:ss-subconn reduc,ss-SP 2vs4 reduc,colDistOracle reduc}, 	
	by running $\cA$ on a graph with $n_0 = \Theta(m^\delta + m^{1-\delta})$ vertices and $m_0 = O(m)$ edges,
	where $m$ is such that $m^\delta = n_1$ and so $m^{1-\delta} = n_2$
	(note that, indeed, $m_0 =  O(\min\{n_0^{1/\delta}, n_0^{1/(1-\delta)}\})$ ),
	$\cB$ has preprocessing time $\poly(m)$ 
	and computation time $O( m^{1-\delta} u(m) + m^\delta q(m) ) = \oo(m)$ which contradicts \Cref{oMv hard} by \Cref{uMv hard}.
	
	The argument for fully dynamic algorithm is similar as in the proof of \Cref{cor:high query}.
\end{proof}

These results show that improving the approximation ratio of 3 of vertex-color distance oracle will cost too much; i.e., we will need $\Omega(n)$ update or query time in a dense graph assuming \Cref{oMv hard} by setting $\delta = 1/2$ and $m=n^2$. 
In particular, one cannot improve the approximation ratio of dynamic Steiner tree with sub-linear update time by improving the approximation ratio of vertex-color distance oracle.

\subsection{Lower Bounds for Graph Problems with other Parameters}\label{sec:lb_graph_other}

\paragraph{$(3-\epsilon)$-approximate $s$-$t$ Shortest Path ($(3-\epsilon)$st-SP) }

By subdividing edges, we can get a weaker lower bound, but better
approximation factor, for distance related problems.
\begin{lemma}
\label{(3-eps) st-SP reduc}Given a partially dynamic algorithm $\cA$
for $(3-\epsilon)$-approximate st-SP, one can solve $\uMv$ with parameters $n_1$ and $n_2$ 
by running the preprocessing step of $\cA$ on a graph with $O(n)$ vertices and then
making $O(\sqrt{n})$ insertions (or $O(\sqrt{n})$ deletions) and $1$ query,
where $n$ is such that $n_1= n_2=\sqrt{n}$.
\end{lemma}
\begin{proof}
We only prove the decremental case, because the incremental case is
symmetric.  
Given~$M$, we construct the bipartite graph $G_{M}$ and add to it vertices $s,t$, and edges $(t,l_{i}),(r_{j},s)$ for all $r_{j}\in R,l_{i}\in L$.
Furthermore we replace each edge $e$ in $G_{M}$ by a path $P_{e}$
of length $\frac{4}{\epsilon}$. Thus, the total number of vertices is
at most $O(n_1 n_2 /\epsilon) = O(n/\epsilon)$. 

Once $u$ and $v$ arrive, we delete $(t,l_{i})$ iff $u_{i}=0$ and
delete $(r_{j},s)$ iff $v_{j}=0$. If $u^{\top}Mv=1$, then $d(s,t)=2+\frac{4}{\epsilon}$,
otherwise $d(s,t)\ge2+3\cdot\frac{4}{\epsilon}$. One can verify that
$\frac{2+3\cdot\frac{4}{\epsilon}}{2+\frac{4}{\epsilon}}>3-\epsilon$ for any $\epsilon>0$.
In total, we need to do $n_1 + n_2 = O(\sqrt{n})$ updates and $1$ query.\end{proof}
\begin{corr}
\label{corr:3approx stSP}
Unless \Cref{oMv hard} fails, there is no partially dynamic algorithm
for $(3-\epsilon)$ st-SP on a graph with $n$ vertices 
with preprocessing time $p(n)=poly(n)$,
worst-case update time $u(n)=\oo(\sqrt{n})$, and query time $q(n)=\oo(n)$
that has an error probability of at most $1/3$.
Moreover, this is true also for fully dynamic algorithm with amortized update time.
\end{corr}

\paragraph{$d$-failure Connectivity }

$d$-failure connectivity problem is a ``1-batch-update'' version of dynamic subgraph connectivity.
The update, for turning off up to $d$ vertices, comes in
one batch. Then one can query whether two nodes $s$ and $t$ are connected.
We want the update time $u(d)$ for a batch of size $d$ and the query time $q(d)$ to depend
mainly on $d$. 

\begin{lemma}
	\label{d-fail conn reduc}
	Let $\delta\in(0,1/2]$ be a fixed constant.
	Given an algorithm $\cA$ for $d$-failure
	connectivity, one can solve $(\frac{\delta}{1-\delta})$-$\uMv$ with parameters $n_1$ and $n_2$ 
	by running the preprocessing step of $\cA$ on a graph with $O(m)$ edges and $\Theta(m^{1-\delta})$ vertices,
	then making $1$ batch of $O(m^\delta)$ updates and $O(m^{1-\delta})$ queries,
	where $m$ is such that $m^{1-\delta} = n_1$ (so $m^\delta = n_2$).
	\end{lemma}
\begin{proof}
Given $M$, we construct the bipartite graph $G_{M}$ and add to it a vertex $s$ and edges $(r_{j},s)$ for all $r_{j}\in R$. There
are $n_1n_2=O(m)$ edges and $n_1+n_2=\Theta(m^\delta + m^{1-\delta}) = \Theta(m^{1-\delta})$ vertices.

Once $u$ and $v$ arrive, we turn off all $r_{j}$ where $v_{j}=0$
in one batch of  $O(m^\delta)$ updates. $u^{\top}Mv=1$ iff, for some $i$, $u_{i}=1$,
$s$ is connected to $u_{i}$. This can be checked using at most $m^{1-\delta}$ queries.\end{proof}
\begin{corr}
	\label{corr:d failure d poly}
	For any $n$, $m = O(n^{1/(1-\delta)})$, and constant $\delta\in(0,1/2]$, 
	unless \Cref{oMv hard} fails, there is no algorithm for $d$-failure
	connectivity for a graph with $n$ vertices and at most $m$ edges
	with preprocessing time $p(n)=poly(n)$, update time $u(d)=\oo(d^{1/\delta})$, and query time $q(d)=\oo(d)$
	that has an error probability of at most $1/3$,
	when $d=m^{\delta}$.
	
	\end{corr}
\begin{proof}
	Suppose there is such an algorithm $\cA$. By \Cref{d-fail conn reduc}, we can solve $(\frac{\delta}{1-\delta})$-$\uMv$ with parameters  $n_1=m^{1-\delta}$ and $n_2=m^\delta$ by running $\cA$ in time $O(u(m^\delta) + m^{1-\delta}q(m^\delta) ) = \oo(m)$. This contradicts \Cref{oMv hard} by \Cref{uMv hard}.
\end{proof}
\Cref{corr:d failure d poly} implies that Duan and Pettie's result \cite{DuanP10}  with preprocessing time 
$\tilde{O}(d^{1-2/c}mn^{1/c - 1/(c\log(2d))})$,
update time $\tilde{O}(d^{2c+4})$ and query time $O(d)$, for any integer $c\ge 1$, is tight in the sense that we cannot improve the query time significantly as long as we want to have update time polynomial in $d$ (because we can choose $\delta$ to be any possible constant close to zero).
However, improving the update time does not contradict the \oMv conjecture.

With the same argument we also get the following lower bound.
\begin{corr}
	\label{corr:d failure dn}
	For any $n$, $m = \Theta(n^{1/(1-\delta)})$, and constant $\delta\in(0,1/2]$, 
	unless \Cref{oMv hard} fails, there is no algorithm for $d$-failure
	connectivity for a graph with $n$ vertices and at most $m$ edges
	with preprocessing time $p(n)=poly(n)$, update time $u(n,d)=\oo(dn)$, and query time $q(d)=\oo(d)$
	that has an error probability of at most $1/3$,
	when $d=m^{\delta}$.	
\end{corr}

\subsection{Lower Bounds for Non-graph Problems}\label{sec:lb_non_graph}
In this section, we mimic the reductions from \cite{Patrascu10} to show hardness of non-graph problems. 
However, our results imply amortized lower bounds while the results in \cite{Patrascu10} are for worst-case lower bounds.

\paragraph{Pagh's problem.}

\begin{lemma}
	\label{pagh reduc in short}Given an algorithm $\cA$ for Pagh's problem (cf. \Cref{table:problem definitions 3}),
	for any constant $\gamma>0$, one can solve $\guMv$ with parameters $k$ and $n$ by running the preprocessing step of $\cA$ on $k$ initial sets in the universe $[n]$
	using $\poly(k,n)$ preprocessing time, $k$ updates 
	(adding $k$ more sets into the family) and $n$ queries.\end{lemma}
\begin{proof}
	Given a matrix $M$, let $\bar{M}$ be the matrix defined by $\bar{M}_{i,j}=1-M_{i,j}$ for all $ i $ and $ j $.
	Let $M_i$ be the $i$-th row of $M$ and treat it as a subset of $[n]$ i.e., $j\in M_i$ iff $M_{i,j}=1$.
	We similarly treat the $i$-th row $\bar{M}_i$ of $\bar{M}$ as a set.
	Note that $u^{\top}Me_{j}=1$
	iff $j\in\bigcup_{i:u_{i}=1}M_{i}$ iff $j\notin\bigcap_{i:u_{i}=1}\bar{M}_{i}$. 
	Thus, at the beginning, we compute $\bar{M}_{i}$ for each $i\le k$. 
	Once $u$ and $v$ arrive, we compute $\bigcap_{u_{i}=1}\bar{M}_{i}$ using
	$k$ updates. There is a $j$ with $v_{j}=1$ such that $j\notin\bigcap_{u_{i}=1}\bar{M}_{i}$
	iff $u^{\top}Mv=1$. We need $n$ queries to check if such a $ j $ exists.\end{proof}

\begin{corr}
	\label{corr:pagh}
	For any constant $\gamma>0$,
	\Cref{oMv hard} implies that there is no dynamic algorithm $\cA$
	for Pagh's problem maintaining $k$ sets over the universe $[n]$ where $k=n^\gamma$,
	with $\poly(k,n)$ preprocessing
	time, $u(k,n)=\oo(n)$ amortized update time, and $q(k,n)=\oo(k)$ query time
	that has an error probability of at most $1/3$.
\end{corr}

\begin{proof}
	Recall that, by the notion of amortization, 
	if $\cA$ initially maintains $k$ sets, then the total update time of $\cA$ is $O((t+k)\cdot u(t+k,n))$.
	By \Cref{pagh reduc in short}, we can solve $\uMv$ with parameters $k$ and $n$ 
	by running $\cA$ and making $k$ updates and $n$ queries
	in time $O(2k \cdot u(2k,n) + n \cdot q(2k,n)) = \oo(nk)$.
	This contradicts \Cref{oMv hard} by \Cref{uMv hard}.
\end{proof}

\paragraph{Langerman's Zero Prefix Sum problem }
\begin{lemma}
	\label{langerman reduc}Given an algorithm $\cA$ for Langerman's problem (cf. \Cref{table:problem definitions 3}),
	one can solve 1-$\uMv$ with parameters $n_1$ and $n_2$ by
	running the preprocessing step of $\cA$ on an array of size $O(n)$, and then
	making $O(\sqrt{n})$ updates and $O(\sqrt{n})$ queries
	where $n$ is such that $n_1=n_2=\sqrt{n}$
	\end{lemma}
\begin{proof}
	Given a matrix $M$, we construct an array $R$ of size $1+n_{1}\cdot(2n_{2}+2)=O(n)$.
	For convenience, we will imagine that $R$ is arranged as a two-dimensional array $\{R_{i,j}\}_{i\in[n_{1}],j\in[2n_{2}+2]}$
	($R_{i,2n_{2}+2}$ is before $R_{i+1,1}$) with one additional entry
	$R_{0}$ at the beginning. 
	
	For all $i \le n_1$, we set $R_{i,1}=0$ and $R_{i,2n_{2}+2}=-2n_{2}$. For each entry
	of $M_{i,j}$, if $M_{i,j}=1$, then we set $R_{i,2j}=1$ and $R_{i,2j+1}=1$. If
	$M_{i,j}=0$, then we set $R_{i,2j}=2$ and $R_{i,2j+1}=0$. Note that $\sum_{j}R_{i,j}=0$
	for all $i$, so the rows of $R$ are ``independent''.
	
	Once $u$ and $v$ arrive, we swap the values of $R_{i,1}$ and $R_{i,2n_{2}+2}$ 
	for all $i$ where $u_{i}=1$ by setting $R_{i,1}=-2n_{2}$ and $R_{i,2n_{2}+2}=0$. For each $j$ where $v_{j}=1$, we
	set $R_{0}=2(n_{2}-j)+1$ and query for a zero prefix sum. See \Cref{fig:langerman}.
	In total, we need to do $O(n_{1}+n_{2})=O(\sqrt{n})$ updates and $O(n_{2})=O(\sqrt{n})$ queries.
	
	To show correctness, we claim that a zero prefix sum exists iff $u^{\top}Me_{j}=1$ where $e_{j}$
	has $1$ at only the $j$-th entry. 
	First, the prefix sums
	cannot reach zero at row $i$ of $R$ if $u_{i}=0$, because in that
	row $i$, each number is positive except $R_{i,2n_{2}+2}=-2n_{2}$
	which just resets the sum within the row to zero. Second, for each row
	$i$ where $u_{i}=1$, the prefix sum from $R_{0}$ to $R_{i,1}$
	is $-2j+1$. Then each pair of entries in the row increments the sum
	by $2$. The prefix sum reaches zero iff $M_{i,j}=1$. If $M_{i,j}=0$,
	then $R_{i,2j}=2$ so the prefix sums to $R_{i,2j-1}$ and $R_{i,2j}$ are $-1$ and $1$, respectively.
	The prefix sum then stays
	positive until row $i$ finishes. If $M_{i,j}=1$, then $R_{i,2j}=1$ and the prefix sum from
	$R_{0}$ to $R_{i,2j}$ is exactly 0.
\end{proof}

\begin{figure}

	\centering
	\begin{tabular}{c|c|cc|cc|c|cc|c|}
		\hline 
		\multicolumn{1}{|c|}{$2(n_{2}-j)+1$} & $0$ & 1 & 1 & 2 & 0 & $\cdots$ & 2 & 0 & $-2n_{2}$\tabularnewline
		\hline 
		& $0$ & 1 & 1 & 2 & 0 & $\cdots$ & 1 & 1 & $-2n_{2}$\tabularnewline
		\cline{2-10} 
		& $-2n_{2}$ & 2 & 0 & 2 & 0 & $\cdots$ & 1 & 1 & 0\tabularnewline
		\cline{2-10} 
		& $-2n_{2}$ & 1 & 1 & 2 & 0 & $\cdots$ & 2 & 0 & 0\tabularnewline
		\cline{2-10} 
		& $0$ & 2 & 0 & 1 & 1 & $\cdots$ & 2 & 0 & $-2n_{2}$\tabularnewline
		\cline{2-10} 
	\end{tabular}		
	\protect\caption{The array in the reduction from $\uMv$ to Langerman's zero prefix
		sum problem}\label{fig:langerman}
\end{figure}

\begin{corr}
	\label{corr:langerman}
	\Cref{oMv hard} implies that there is no algorithm for Langerman's problem on an array of size $n$
	with preprocessing time $p(n)=poly(n)$, 
	amortized update time $u(n)=\oo(\sqrt{n})$, and query time $q(m)=\oo(\sqrt{n})$
	that has an error probability of at most $1/3$.
\end{corr}

\begin{proof}
	Suppose there is such an algorithm $\cA$. By \Cref{langerman reduc} and by ``resetting'' the array when the new vector pair arrives, one can solve $\ouMv$ with parameters $n_1=\sqrt{n}$, $n_2=\sqrt{n}$, and $n_3=\sqrt{n}$ in time $O(n\cdot u(n) + n\cdot q(n)) = \oo(n\sqrt{n})$. 
	This contradicts \Cref{oMv hard} by \Cref{ouMv hard}. Note that by the choice of $ n_3 = \sqrt{n} $ as the third parameter of $\ouMv$ we perform $ \Theta (n) $ updates to $\cA$, which allows use to use the amortized update time $ u(n) $ of $\cA$ in this argument.
\end{proof}

\paragraph{Erickson's problem}

\begin{lemma}
	\label{lem:erikson reduc}
	Given an algorithm $\cA$ for Erickson's problem (cf. \Cref{table:problem definitions 3}), one can
	solve 1-$\ouMv$ with parameters $n_1$,$n_2$, and $n_3$ by running the preprocessing step of $\cA$
	on a matrix of size $n\times n$ and then making 
	$O(n\cdot n_{3})$ updates and $n_{3}$ queries,
	where $n$ is such that $n_1 = n_2 = n$.
	\end{lemma}
\begin{proof}
	Given a Boolean matrix $M$, $\cA$ runs on the same matrix but treats it as an integer matrix.	
	Once $u^{t}$ and $v^{t}$ arrive, we increment the row $i$ iff $u_{i}^{t}=1$
	and increment the column~$j$ iff $v_{j}^{t}=1$. Before $u^{t+1}$
	and $v^{t+1}$ arrive, we increment the remaining rows $i$ where $u_{i}^{t}=0$
	and remaining column where $v_{j}^{t}=0$.
	Therefore, we have that $(u^{t})^{\top}Mv^{t}=1$ iff the maximum value in the matrix is $2t+1$.
\end{proof}

\begin{corr}
	\label{corr:erickson}
	Unless \Cref{oMv hard} fails, there is no algorithm
	for Erickson's problem on a matrix of size $n\times n$
	with preprocessing time $p(n)=poly(n)$, amortized
	update time $u(n)=\oo(n)$, and query time $q(n)=\oo(n^{2})$
	that has an error probability of at most $1/3$.
\end{corr}
\begin{proof}
	Otherwise, we can solve 1-$\ouMv$ with parameters $n_1=n$, $n_2=n$, and $n_3$ using polynomial preprocessing time and $\oo(n_1 n_2 n_3)$ computation time which contradicts \Cref{oMv hard} by \Cref{ouMv hard}.
\end{proof}

\subsection{$(2-\epsilon)$ Approximate Diameter on Weighted Graphs}\label{sec:lb_diam}

We show that, for any $ \epsilon > 0 $, it is $\oMv$-hard to maintain a $ (2-\epsilon) $-approximation of the diameter in a weighted graph under both insertions and deletions with $\oo(\sqrt{n})$ update time and $\oo(n)$ query time, even if the edge weights are only $ 0 $ and $ 1 $.
This reduction is inspired by a lower bound in distributed computation~\cite{FrischknechtHW12}. 
It is different from previous reductions in that we can show the hardness for this problem only in the fully dynamic setting and not in the partially dynamic setting.

\begin{lemma}
	\label{(2-eps) diam reduc}
	For any  $\gamma>0$,
	Given a fully dynamic algorithm $\cA$ for
	$(2-\epsilon)$-approximate diameter on a $\{0,1\}$-weighted undirected graph, one can solve $\guMv$ with parameters $n_1$ and $n_2$
	by running the preprocessing step of $\cA$ on a graph with $O(n_{1}\sqrt{n_{2}})$	vertices,
	and then making $n_{2}+O(n_{1}\sqrt{n_{2}})$ updates and $n_{1}$ queries to $\cA$.
\end{lemma}
\begin{proof}
	First, let us define an undirected graph $H^{v}$, called \emph{vector graph}, from a
	vector $v$ of size $n_{2}$. A vector graph $H^{v}=(B^{v}\cup C^{v})$
	has two halves, called \emph{upper} and \emph{lower halves} denoted by $B^v$ and $C^v$ respectively. 
	$B^{v}$ and $C^{v}$
	are both cliques of size $\sqrt{n_{2}}$. Let $\{b_{x}^{v}\}_{1 \leq x \leq \sqrt{n_{2}}}$
	and $\{c_{x}^{v}\}_{1 \leq x \leq \sqrt{n_{2}}}$ be the vertices of $B^{v}$ and $C^{v}$
	respectively. We include an edge $(b_{x}^{v},c_{y}^{v})$ iff $v_{(x-1)\sqrt{n_{2}}+y}=0$.
	The weight of all edges in $H^{v}$ is $1$.
	
	Given a matrix $M$ of size $n_{1}\times n_{2}$, let $M_{i}$
	be the $i$-th row of $M$. We construct a vector graph $H^{M_{i}}$
	for each $1 \leq i \leq n_1 $, and another vector graph $H^{v}$ where $v$ is the zero vector. 
	There are two special vertices $a$ and $z$. Connect $a$ to all vertices in $H^{v}$ with weight
	one. Connect $z$ to $a$ and all vertices in $H^{M_{i}}$ with weight
	zero, for every $1 \leq i \leq n_1 $.
	
	Once $(u,v)$ arrives, we update $H^{v}$ to be the vector graph of $v$
	in $n_{2}$ updates. Then we work in stages $i$, for each $i$ where
	$u_{i}=1$. Before going to the next stage, we undo all the updates. 
	In stage $i$, 1) disconnect $z$ from each vertex in $H^{M_{i}}$, 2) connect
	$a$ to each vertex in $H^{M_{i}}$ with weight one, and 3) add $0$-weight
	\emph{matching edges} $(b_{x}^{M_{i}},b_{x}^{v})$ and $(c_{y}^{M_{i}},c_{y}^{v})$
	for all $x,y\le\sqrt{n_{2}}$. 
	All three steps need $O(\sqrt{n_{2}})$
	updates. Let $G(M_i,v)$ denote the resulting graph (see \Cref{fig:diam} for example). We query the diameter of $G(M_i,v)$ and will use the result to solve $\guMv$.
	After finishing all stages, there are $n_{2}+O(n_{1}\sqrt{n_{2}})$ updates and $n_1$ queries in total.
	
	\begin{figure}
		\centering
		\includegraphics[width=0.7\textwidth]{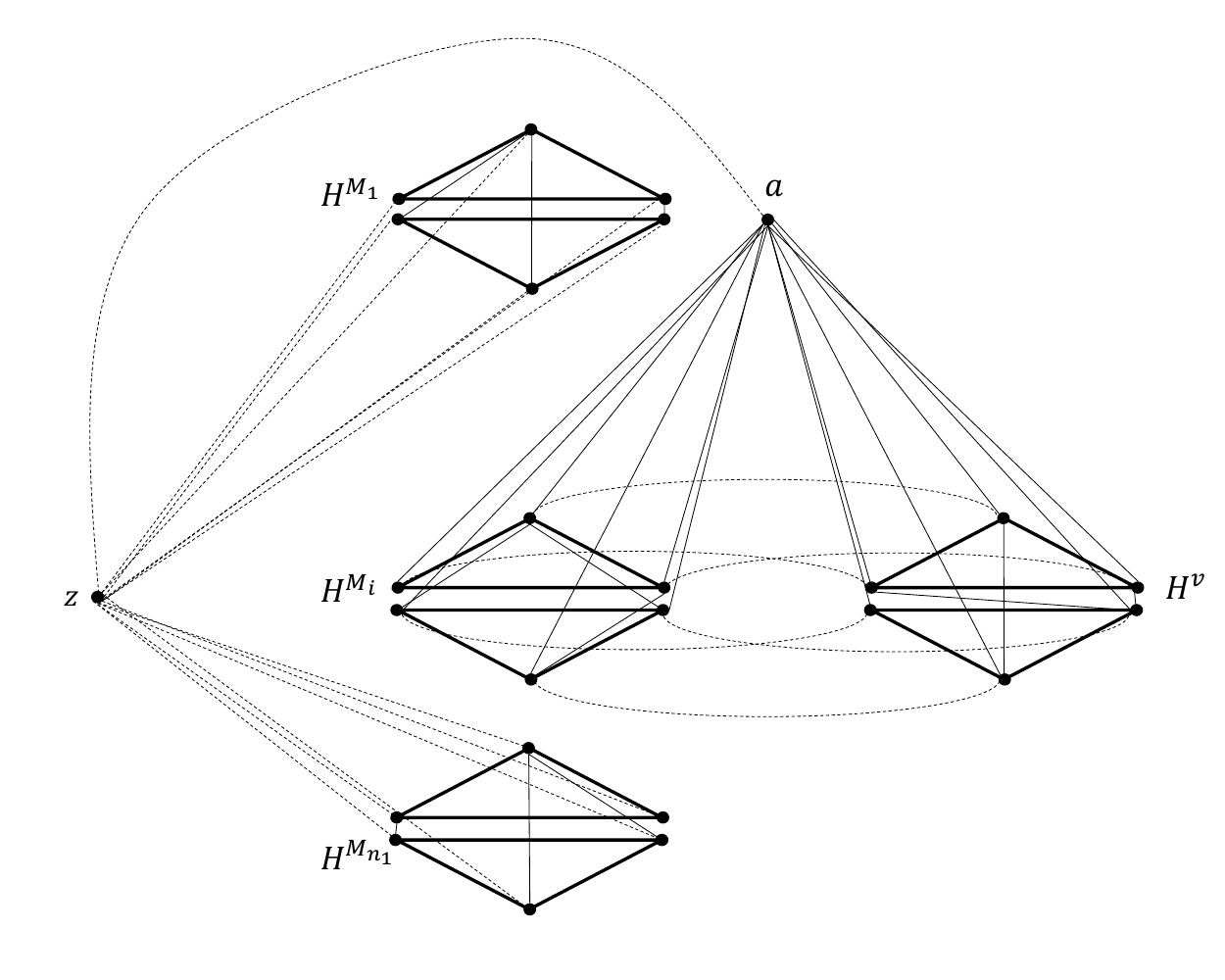}
		\caption{An example of $G(M_i,v)$ from the proof of \Cref{(2-eps) diam reduc}. Dashed lines are edges of weight zero. Other lines are edges of weight one.}				
		\label{fig:diam} 
	\end{figure}
	
	If there is some stage $i$ where the diameter of the graph  $G(M_i,v)$ is $2$, then report $u^{\top}Mv=1$.
	Otherwise report $ 0 $. The following claim justifies this answer.
	\begin{claim}
		The diameter of  $G(M_i,v)$ is $1$ if $M_{i}^{\top}v=0$. Otherwise,
		the diameter is $2$. \end{claim}
	\begin{proof}
		First, every edge incident to $z$ has weight $0$. So we can treat
		all the adjacent vertices of $z$, which are exactly $a$ and those
		in $H^{M_{i'}}$ where $i'\neq i$, as a single vertex. Therefore,
		we just have to analyze the distance among vertices in $H^{M_{i}}$,
		$H^{v}$ and the vertex $a$.
		
		Note that $a$ is connected to all vertices in $H^{M_{i}}$
		and $H^{v}$ by $1$-weight edges. The distance between vertices among the upper halves
		$B^{M_{i}}$ and $B^{v}$ is at most $1$, because $B^{M_{i}}$ and
		$B^{v}$ are cliques and there are matching edges of weight zero. Similarly,
		for the lower halves $C^{M_{i}}$ and $C^{v}$. Now, we are left with
		analyzing the distance between a vertex in the upper halves and another
		vertex in the lower ones. There are two cases. 
		
		If $M_{i}^{\top}v=0$, then for each $x,y\le\sqrt{n_{2}}$, there
		is either an edge $(b_{x}^{M_{i}},c_{y}^{M_{i}})$ or an edge $(b_{x}^{v},c_{y}^{v})$.
		So given any $b_{x}^{M_{i}}$ and $c_{y}^{M_{i}}$, there is either
		a path $(b_{x}^{M_{i}},b_{x}^{v},c_{y}^{v},c_{y}^{M_{i}})$ or a path
		$(b_{x}^{M_{i}},c_{y}^{M_{i}})$ both of weight 1. Since $d(b_{x}^{M_{i}},c_{y}^{M_{i}})=1$,
		the distance among $b_{x}^{M_{i}},b_{x}^{v},c_{y}^{M_{i}},c_{y}^{v}$
		is at most 1. Since this is true for any $x,y$, the diameter of the
		graph is $1$.
		
		If $M_{i}^{\top}v=1$, then there are some $x,y$ such that neither
		the edge $(b_{x}^{M_{i}},c_{y}^{M_{i}})$ nor the edge $(b_{x}^{v},c_{y}^{v})$ exists. To show
		that $d(b_{x}^{M_{i}},c_{y}^{M_{i}})\ge2$, it is enough to show that $d(b_{x}^{M_{i}},c_{y}^{M_{i}})>1$
		because every non-zero edge weight is $1$. The set of vertices with
		distance $1$ from $b_{x}^{M_{i}}$ includes exactly the neighbors
		of $b_{x}^{M_{i}}$ in $H^{M_{i}}$ and their ``matching'' neighbors in $ H^v $, the neighbors of $b_{x}^{v}$
		in $H^{v}$ and their ``matching'' neighbors in $ H^{M_i} $, and the vertex $a$, which combines $z$ and all vertices
		in $H^{M_{i'}}$ where $i'\neq i$. But this set does not include
		$c_{y}^{M_{i}}$. Therefore $d(b_{x}^{M_{i}},c_{y}^{M_{i}})\ge2$
		and we are done.
	\end{proof}
\end{proof}
\begin{corr}
	\label{corr:approx diam}
	Assuming \Cref{oMv hard}, there is no fully dynamic algorithm
	for $(2-\epsilon)$-approximate diameter on $\{0,1\}$-weighted graphs with $n$ vertices
	with preprocessing time $p(n)=poly(n)$, amortized update time $u(n)=\oo(\sqrt{n})$,
	and query time $q(n)=\oo(n)$ 
	that has an error probability of at most $1/3$.
\end{corr}
\begin{proof}
	Suppose such a fully dynamic algorithm $\cA$ exists. By \Cref{(2-eps) diam reduc} and by ``undoing'' the operations as in the proof of \Cref{cor:high query},
	we can solve $2$-$\ouMv$ with parameters $n_1=\sqrt{n}$, $n_2=n$, and $n_3=n$ by running $\cA$ on a graph with $\Theta(n_{1}\sqrt{n_{2}})=\Theta(n)$ vertices, and then making $O(n_{2}+n_{1}\sqrt{n_{2}})\times n_3 = O(n^2)$ updates and $n_1 n_3 = O(n\sqrt{n})$ queries in total. 
	The computation time is $O(n^2 u(n)+ n\sqrt{n} q(n)) = \oo(n^2\sqrt{n})$, contradicting \Cref{oMv hard} by \Cref{ouMv hard}.
	Note that we choose $n_3 = n\sqrt{n}$ to make sure that the number of updates is at least the number of edges in the graph, so that we can use the amortized time bound.
\end{proof}

\subsection{Densest Subgraph Problem}\label{sec:lb_densest}

In this section, we show a non-trivial reduction from 1-$\uMv$ to the densest subgraph problem and hence show the hardness of this problem.
\begin{theorem}
	\label{thm:densest reduc}
	Given a partially dynamic $\cA$ for maintaining the density of the densest subgraph,
	one can solve 1-$\uMv$ with parameters $n_1=n$ and $n_2=n$ by running the preprocessing step of $\cA$ on a graph with $\Theta(n^3)$ vertices, 
	and then making $\Theta(n)$ updates and 1 query.
\end{theorem}

\paragraph{Problem definition}

We are given an undirected input graph $G=(V,E)$ with vertices $ V $
and edges $ E $. For every subset of vertices $S\subseteq V$, let $G(S)=(S,E(S))$
denote the subgraph of~$G$ induced by the vertices in $S$, i.e., we
have $E(S)=\{(u,v)\in E \mid u,v\in S\}$. The \emph{density} of any subset of vertices $S\subseteq V$
is defined as $\rho(S)=|E(S)|/|S|$.
For the reduction described in the following let $M$ be a Boolean matrix of size
$n\times n$ and set $ k = 6n $.

\paragraph{Preprocessing.} 
We construct the graph $ G $ as follows:
\begin{itemize}
\item \emph{Bit graphs for $M$.} For each bit $m_{i,j}$ of $M$, construct
a graph $B_{i,j}$ consisting of $k$ vertices.
There are two special vertices in $B_{i,j}$,
called \emph{special vertex~1} and \emph{special vertex~2}.
If the bit $ m_{i,j} $ is set, connect the nodes in $B_{i,j}$ by a path of $k-1$ edges in $B_{i,j}$ from special vertex 1 to special vertex 2.
If the bit $ m_{i,j} $ is not set, insert no edges into $B_{i,j}$. 
\item \emph{Row graph for M.} For each row $i$ of $M$, construct a graph
$R_{i}$ consisting of $3$ vertices. One of these vertices is special. Add
an edge from the special vertex of $R_{i}$ to special vertex 1 of $B_{i,j}$
for all $1\le j\le n$.
\item \emph{Column graph for $M$.} For each column $j$ of $M$, construct
a graph $C_{j}$ consisting of $3$ vertices. One of these vertices is special.
Add an edge from the special vertex of $C_{j}$ to special vertex 2 of
$B_{i,j}$ for all $1\le i\le n$.
\end{itemize}
Observe that $ G $ has $ O (n^3) $ vertices.

\paragraph{Revealing $u$ and $v$}
We execute the following $O(n)$ edge operations and one query.
\begin{itemize}
\item For each $i$ where $u_{i}=1$, turn the row graph $R_{i}$ into a
triangle by inserting $ O(n) $ edges.
\item For each $j$ where $v_{j}=1$, turn the column graph $C_{i}$ into
a triangle by inserting $ O(n) $ edges.
\item Then ask for the size of the densest subgraph.
\end{itemize}

This describes the reduction of \Cref{thm:densest reduc}. Note that we only need a partially dynamic algorithm.
Before proving the correctness of this reduction below, we observe the following easy lemma.

\begin{lemma}
\label{lem:ratio}
For all numbers $a$, $b$, $c$, $d$, and $r$ we have:
\begin{enumerate}
\item If $ \frac{a}{b} \geq r $ and $ \frac{c}{d} \geq r $, then $ \frac{a + c}{b + d} \geq r $.
\item If $ \frac{a}{b} \geq r $ and $ \frac{c}{d} \leq r $, then $ \frac{a - c}{b - d} \geq r $.
\end{enumerate}
\end{lemma}

\begin{theorem}
\label{thm:density reduc correct}There exists a subset $ S \subseteq V $ with density $ \rho (S) \geq \frac{k+7}{k+6}$ if and only if $u^{\top}Mv=1$.
\end{theorem}

\begin{proof}
Assume first that $u^{\top}Mv=1$. So there are indices $i$ and $j$ such that $u_{i}M_{i,j}v_{j}=1$.
Consider the subgraph consisting of the union
of $R_{i}$, $B_{i,j}$ and $C_{j}$. It consists of $ k + 6 $ vertices and
$ k + 7 $ edges, i.e., it has density $\frac{k+7}{k+6}$.

Now assume that $u^{\top}Mv=0$ and let $ S \subseteq V $.
To show that $ \rho (S) \leq \frac{k+7}{k+6} $ we will make the following assumptions:
\begin{enumerate}[label=(\arabic*)]
\item For every $ i $ and $ j $, either the full bit graph $ B_{i,j} $ together with two edges leaving the bit graph is contained in $ S $ or no node of $ B_{i,j} $ is contained in $ S $.
\item For every row $ i $ of a set bit (i.e., where $ u_i = 1 $), either the full row graph $ R_i $ is contained in $ G(S) $ or no node of $ R_i $ is contained in $ G(S) $
\item For every row $ i $ of an unset bit (i.e., where $ u_i = 0 $), either the special node of the row graph $ R_i $ is contained in $ S $ or no node of $ R_i $ is contained in $ G(S) $
\item For every column $ j $ of a set bit (i.e., where $ v_j = 1 $), either the full column graph $ R_j $ is contained in $ G(S) $ or no node of $ C_j $ is contained in $ G(S) $
\item For every column $ j $ of an unset bit (i.e., where $ v_j = 0 $), either the special node of the column graph $ C_j $ is contained in $ S $ or no node of $ C_j $ is contained in $ G(S) $
\end{enumerate}
These assumptions can be made without loss of generality as we argue in the following.

(1) Suppose $ S $ contains some subset $ U $ of nodes of $ B_{i,j} $ and either $ U $ does not contain all nodes of $ B_{i,j} $ or one of the special nodes of $ B_{i,j} $ is not contained in $ S $.
Then by removing $ U $ from $ S $ we remove some $ q \leq k $ nodes and at most $ q $ edges from $ G(S) $.
Thus, we are removing a piece of density at most $ 1 $.
If $ \rho (S) \geq \frac{k+7}{k+6} \geq 1 $, then removing $ U $ from $ S $ will not decrease $ \rho (S) $ by Part~2 of \Cref{lem:ratio} (using $ r = \frac{k+7}{k+6} $, $ \rho(S) = \frac{a}{b} $, $ c \leq q $, and $ d = q $).
Therefore we may assume without loss of generality that $ S $ does not contain $ U $.

(2) If only one of the nodes of $ R_i $ is contained in $ S $, then by adding the two other nodes we add $ 2 $ nodes and $ 3 $ edges to $ G(S) $.
As $ \frac{3}{2} > \frac{k+7}{k+6} $, doing so will not decrease the density of $ S $ to below $ \frac{k+7}{k+6} $ by Part~1 of \Cref{lem:ratio} (using $ r = \frac{k+7}{k+6} $, $ \rho(S) = \frac{a}{b} $, $ c = 3 $, and $ d = 2 $) and thus we may assume without loss of generality that all the nodes of $ R_i $ are contained in $ S $.
Similarly, if only two of the nodes of $ R_i $ are contained in $ S $, then by adding the third node we add $ 1 $ node and $ 2 $ edges to $ G(S) $.
Again, by \Cref{lem:ratio}, we may assume without loss of generality that all three nodes of $ R_i $ are contained in $ S $.

(3) There are no edges incident to the non-special nodes of $ R_i $.
By removing the non-special nodes of $ R_i $ from $ S $ we only increase $ \rho (S) $.
Thus, we may assume without loss of generality, that only the special node is contained in $ S $.

(4) and (5) follow the same arguments as (2) and (3).

Using these assumptions we conclude that $ G(S) $ has the following structure: it contains some full bit graphs (i.e., paths) of set bits, each with two outgoing edges, some full row or column graphs (i.e., triangles) of set bits, and some special nodes of row or column graphs of unset bits.
In the rest of this proof we will use the following notation:
$ x $ denotes the number of bit graphs contained in $ S $, $ y $ denotes the number of row or column graphs of set bits contained in $ S $, and $ z $ denotes the number of row or column graphs of unset bits contained in $ C $.
Thus, $ G(S) $ has the density
\begin{equation*}
\rho (S) = \frac{3 y + (k+1) x}{3 y + z + k x} \, .
\end{equation*}
The inequality $ \rho (S) < \frac{k+7}{k+6} $, which we want to prove, is now equivalent to
\begin{equation*}
6 x < (k+7) z + 3 y \, .
\end{equation*}
Consider some bit graph contained in $ G(S) $.
As argued above, this graph has one edge going to a row graph and one edge going to a column graph.
As $u^{\top}Mv=0$ at least one of those edges must go to a row or column graph of an unset bit.
In this way we assign at most $ n $ bit graphs to every unset row or column bit and it follows that $ x \leq n z $.
As we have defined $ k = 6n $ we obtain
\begin{equation*}
6 x \leq 6 n z = k z < (k+7) z + 3 y
\end{equation*}
as desired.
\end{proof}

This complete the proof of \Cref{thm:densest reduc}.

\begin{corollary}
\label{corr:densest}
Unless \Cref{oMv hard} fails, there is no partially dynamic algorithm
$\cA$ for maintaining the density of the densest subgraph on a graph with $n$ vertices
with polynomial preprocessing time, worst-case update
time $u(n)=\oo(n^{1/3})$, and query time $q(n)=\oo(n^{2/3})$
that has an error probability of at most $1/3$.
Moreover, this is true also for fully dynamic algorithms with amortized
update time.
\end{corollary}
\begin{proof}
Suppose that such a partially dynamic algorithm $\cA$ exists.
By \Cref{thm:densest reduc} and by scaling down the parameter from $n$ to $n^{1/3}$,
we can solve 1-$\uMv$ with parameters $n^{1/3}$ and $n^{1/3}$, by running $\cA$ on a graph with $\Theta(n)$ vertices,
in time $O(n^{1/3} u(n) + q(n)) = \oo(n^{2/3})$ contradicting \Cref{oMv hard} by \Cref{uMv hard}.

If $\cA$ is fully dynamic, the argument is similar as in the proof of \Cref{cor:high query}.
\end{proof}

\section{Hardness for Total Update Time of Partially Dynamic Problems}\label{sec:full:partially hardness}

Our lower bounds, compared to previously known bounds, for the total update time of partially dynamic problems are summarized in \Cref{table: summary partially}. Tight results are summarized in \Cref{table: tight results partially}

\begin{table}
\footnotesize
\begin{tabular}{|>{\centering}p{0.3\textwidth}|>{\centering}p{0.08\textwidth}|>{\centering}p{0.08\textwidth}|>{\centering}p{0.08\textwidth}|>{\centering}p{0.08\textwidth}|>{\centering}p{0.11\textwidth}|>{\centering}p{0.22\textwidth}|}
	\hline 
	Problems & $p(m,n)$ & $u(m,n)$ & $q(m,n)$ & Conj.& Reference & Remark\tabularnewline
\hline 
\hline 
\multirow{2}{0.3\textwidth}{Bipartite Max Matching } & $m^{4/3-\epsilon}$ & $m^{4/3-\epsilon}$ & $m^{1/3-\epsilon}$ & 3SUM  & \cite{KopelowitzPP14} & $m=\Theta(n)$; only for incremental case.\tablefootnote{Kopelowitz, Pettie, and Porat~\cite{KopelowitzPP14} show a higher lower bound for the amortized update time of incremental algorithms. But they allow reverting the insertion which is not allowed in our setting.}\tabularnewline
\cline{2-7} 
 & $\mathbf{poly}$  & $\mathbf{m^{3/2-\epsilon}}$  & $\mathbf{m}^{\mathbf{1-\epsilon}}$  & $\oMv$ & \Cref{corr:bipartitie partial} & $m=\Theta(n)$ \tabularnewline
\hline 
unweighted st-SP  & $\mathbf{poly}$  & $\mathbf{m^{3/2-\epsilon}}$  & $\mathbf{m}^{\mathbf{1-\epsilon}}$  & $\oMv$ & 
\Cref{cor:st-SP hard partial} & $m=O(n^2)$ \tabularnewline
\hline 
\multirow{3}{0.3\textwidth}{unweighted ss-SP } & $\mathbf{poly}$  & $\mathbf{m^{3/2-\epsilon}}$  & $\mathbf{m}^{\mathbf{1-\epsilon}}$  & $\oMv$ & \Cref{cor:st-SP hard partial} & $m=O(n^2)$ \tabularnewline
\cline{2-7} 
 & $(mn)^{1-\epsilon}$ & $(mn)^{1-\epsilon}$ & $m^{\delta'-\epsilon}$ & BMM  & \cite{RodittyZESA04} &
 \multirow{6}{0.25\textwidth}{Choose any $\delta' \in (0,1/2]$, $m=\Theta(n^{1/(1-\delta')})$ } 
 \tabularnewline
\cline{2-6} 
 & $\mathbf{poly}$  & $\mathbf{(mn)^{1-\epsilon}}$  & $\mathbf{ m^{\delta'-\epsilon} }$ & $\oMv$ &\Cref{cor:trade-off hardness partial mn} &\tabularnewline
\cline{1-6} 
\multirow{2}{0.3\textwidth}{unweighted $(\alpha,\beta)$ap-SP, $2\alpha+\beta<4$} \\ & $(mn)^{1-\epsilon}$ & $(mn)^{1-\epsilon}$ & $m^{\delta'-\epsilon}$ & BMM & \cite{DorHZ00}  & \tabularnewline
\cline{2-6} 
 & $\mathbf{poly}$  & $\mathbf{(mn)^{1-\epsilon}}$   & $\mathbf{ m^{\delta'-\epsilon} }$ & $\oMv$ & \Cref{cor:trade-off hardness partial mn} &\tabularnewline
\cline{1-6} 
\multirow{2}{0.3\textwidth}{Transitive Closure } & $(mn)^{1-\epsilon}$ & $(mn)^{1-\epsilon}$ & $m^{\delta'-\epsilon}$ & BMM  & \cite{DorHZ00} &\tabularnewline
\cline{2-6} 
 & $\mathbf{poly}$  & $\mathbf{(mn)^{1-\epsilon}}$   & $\mathbf{ m^{\delta'-\epsilon} }$ & $\oMv$ & \Cref{cor:trade-off hardness partial mn} & \tabularnewline
\hline 
\end{tabular}

\caption{Lower bounds for total update time of partially dynamic problems.
Bounds which are not subsumed are highlighted. Each row states that
there is no algorithm achieving stated preprocessing time, total update
time, and query time \emph{simultaneously}, unless the conjecture
fails. Lower bounds based on BMM apply to only combinatorial algorithms.
}

\label{table: summary partially}
\end{table}

\begin{table}
\footnotesize
\begin{tabular}{|>{\centering}p{0.15\textwidth}|>{\centering}p{0.07\textwidth}|>{\centering}p{0.07\textwidth}|>{\centering}p{0.07\textwidth}|>{\centering}p{0.07\textwidth}|>{\centering}p{0.07\textwidth}|>{\centering}p{0.07\textwidth}|>{\centering}p{0.15\textwidth}|>{\centering}p{0.15\textwidth}|}
\hline 
\multirow{2}{0.15\textwidth}{Problem} & \multicolumn{3}{c|}{Upper Bounds} & \multicolumn{3}{c|}{Lower Bounds} & \multirow{2}{0.15\textwidth}{Problem} & \multirow{2}{0.15\textwidth}{Remark}\tabularnewline
\cline{2-7} 
 & $p(m,n)$ & $u(m,n)$ & $q(m,n)$  & $p(m,n)$ & $u(m,n)$ & $q(m,n)$ &  & \tabularnewline
\hline 
\hline 
dec. unweighted Exact ss-SP & $m$ & $mn$ & 1 & \multirow{2}{0.07\textwidth}{$poly$ } & \multirow{2}{0.07\textwidth}{$mn^{1-\epsilon}$ } & \multirow{2}{0.07\textwidth}{$m^{\delta'-\epsilon}$ } & \multirow{2}{0.15\textwidth}{dec. unweighted Exact ss-SP} & Upper: \cite{EvenS81} \tabularnewline
\cline{1-4} \cline{9-9} 
dec. unweighted $(1+\epsilon)$ ss-SP & $m^{1+o(1)}$ & $m^{1+o(1)}$  & 1  &  &  &  &  & Upper: \cite{HenzingerKNFOCS14} \tabularnewline
\hline

dec. unweighted $(1+\epsilon,0)$ ap-SP & $mn$ & $mn$ & 1 & \multirow{3}{0.07\textwidth}{$poly$ } & \multirow{3}{0.07\textwidth}{$mn^{1-\epsilon}$} & \multirow{3}{0.07\textwidth}{$m^{\delta'-\epsilon}$ } & \multirow{3}{0.15\textwidth}{dec. $(\alpha,\beta)$ ap-SP, $2\alpha+\beta<4$} & Upper: \cite{HenzingerKN13,RodittyZ04,Bernstein13} \tabularnewline
\cline{1-4} \cline{9-9} 
dec. unweighted $(1+\epsilon,2)$ ap-SP & $n^{5/2}$ & $n^{5/2}$ & 1 &  &  &  &  & \multirow{2}{0.15\textwidth}{\centering Upper: \cite{HenzingerKN13}}\tabularnewline
\cline{1-4} 
dec. unweighted $(2+\epsilon,0)$ ap-SP & $n^{5/2}$ & $n^{5/2}$ & 1 &  &  &  &  & \tabularnewline
\hline 

dec. Transitive Closure & $mn$ & $mn$ & $1$ & $poly$  & $mn^{1-\epsilon}$ & $m^{\delta'-\epsilon}$  & dec. transitive closure & Upper: \cite{Lacki13}
\tabularnewline
\hline 
\end{tabular}{\footnotesize \par}

\caption{Our tight results along with the matching upper bounds (or better
upper bounds when worse approximation ratio is allowed). The polylogarithmic
factors are omitted. The lower bounds state that there is no algorithm
achieving stated preprocessing time, total update time, and query time
\emph{simultaneously}, unless the conjecture fails. The table shows
that one cannot improve the approximation factor of decremental $(1+\epsilon)$
ss-SP/$(1+\epsilon,2)$ ap-SP/$(2+\epsilon,0)$ ap-SP without scarifying
fast running time. All lower bounds hold, for any $\delta'\in (0,1/2]$, when $m = \Theta(n^{1/(1-\delta')})$. $\epsilon>0$ is any constant.
 }

\label{table: tight results partially}
\end{table}

Given a matrix $M\in\{0,1\}^{n_{1}\times n_{2}}$, we denote a bipartite
graph $G_{M}=((L,R),E)$ where $L=\{l_{1},\dots,l_{n_{1}}\}$, $R=\{r_{1},\dots,r_{n_{2}}\}$,
and $E=\{(r_{j},l_{i})\mid M_{ij}=1\}$.

In this section, our proofs again follow two simple steps as in \Cref{sec:full:fully hardness}.
First, we show the reductions in lemmas that 
given a partially dynamic algorithm $\cA$ for some problem, 
one can solve $\ouMv$ by running the preprocessing step of $\cA$ on some graph 
and then making some number of updates and queries.
Then, we conclude in corollaries that if 
$\cA$ has low total update update time and query time
then this contradicts \Cref{oMv hard}. 

In the proofs of the lemmas of this section, we only usually show the reduction from $\ouMv$ 
to the decremental algorithm, because it is symmetric in the incremental setting.

\paragraph{$s$-$t$ Shortest Path (st-SP)}
\begin{lemma}
\label{lem:st-SP reduc partial}
Given an incremental (respectively decremental) dynamic algorithm $\cA$ for st-SP, 
one can solve 1-$\ouMv$ with parameters $n_1$, $n_2$, and $n_3$ 
by running the preprocessing step of $\cA$ on a graph with $\Theta(\sqrt{m})$ vertices
and $O(m)$ edges
which is initially empty (respectively initially has $\Theta(m)$ edges), 
and then making $\sqrt{m}$ queries,
where $m$ is such that $n_1=n_2=n_3=\sqrt{m}$.
\end{lemma}
\begin{proof}
Given an input matrix $M$ of 1-$\ouMv$, we construct a bipartite graph $G_{M}$, and also two
paths $P$ and $Q$ with $n_{3}$ vertices each. Let $P=(p_{1},p_{2},\dots,p_{n_{3}})$
where $p_{1}=s$, and $Q=(q_{1},q_{2},\dots,q_{n_{3}})$ where $q_{1}=s'$.
Add the edge $(p_{t},l_{i})$ and $(q_{t},r_{j})$ for all $i\le n_1, j\le n_2$ and $t\le n_3$.
There are $\Theta(n_1+n_2+n_3) = \Theta(\sqrt{m})$ vertices and $O(n_1n_2+n_3)=O(m)$ edges.

Once $u^{t}$ and $v^{t}$ arrive, we disconnect $p_{t}$ from $l_{i}$
iff $u_{i}^{t}=0$, and disconnect $q_{t}$ from $r_{j}$ iff $v_{j}^{t}=0$.
We have that if $(u^{t})^{\top}Mv^{t}=1$, then $d(s,s')=2t+1$, otherwise
$d(s,s')\ge2t+2$. This is because, before $u^{t+1}$ and $v^{t+1}$ arrive, we disconnect
$p_{t}$ from $l_{i}$ for all $i\le n_1$ and disconnect $q_{t}$ from $r_{j}$ for all $j\le n_2$. 
So for each $t$, we need 1 query, and hence $n_3=\sqrt{m}$ queries in total.
\end{proof}

\begin{corollary}
	\label{cor:st-SP hard partial}
	For any $n$ and $m = O(n^2)$, \Cref{oMv hard} implies that there is no partially dynamic st-SP algorithm $\cA$ on a graph with $n$ vertices and at most $m$ edges 
	with polynomial preprocessing time, total update time $\oo(m^{3/2})$, and query time $\oo(m)$ that has an error probability of at most $1/3$.
\end{corollary}
\begin{proof}
	Suppose there is such an algorithm $\cA$.
	By \Cref{lem:st-SP reduc partial}, we construct an algorithm $\cB$ for 1-$\ouMv$ with parameters 
	$n_1=\sqrt{m}$, $n_2=\sqrt{m}$, and $n_3=\sqrt{m}$
	by running $\cA$ on a graph with $n_0=\Theta(\sqrt{m})$ vertices and $m_0 = O(m)$ edges.
	Note that $m_0 = O(n_0^2)$.
	Since $\cA$ uses polynomial preprocessing time,
	total update time $\oo(m^{3/2})$
	and total query time $O(\sqrt{m} q(m) )=\oo(m^{3/2})$
	$\cB$ has polynomial preprocessing time and $\oo(m^{3/2})$ computation time contradicting \Cref{oMv hard} by \Cref{ouMv hard}.
\end{proof}

Note that, when $m=\Theta(n^2)$, \Cref{cor:st-SP hard partial} implies that there is no algorithm with $\oo(n^3)$ total update time and $\oo(n^2)$ query time. There is a matching upper bound of total update time $O(m n)=O(n^3)$ due to \cite{EvenS81}.

\paragraph{Bipartite Maximum Matching}
\begin{lemma}
Given an incremental (respectively decremental) dynamic algorithm for bipartite maximum matching, 
one can solve 1-$\ouMv$ with parameters $n_1$, $n_2$, and $n_3$
by running $\cA$ on a graph with $\Theta(m)$ vertices and $\Theta(m)$ edges
which is initially empty (respectively initially has $\Theta(m)$ edges), 
and then making $\sqrt{m}$ queries,
where $m$ is such that $n_1=n_2=n_3=\sqrt{m}$.
\end{lemma}
\begin{proof}
Given an input matrix $M$ of 1-$\ouMv$, we perform the following preprocessing.
First, we construct a bipartite graph $G_{M}$ which has $O(n_1n_2)=O(m)$ edges.
There are also additional sets of vertices 1) $L'=\{l'_1,\dots,l'_{n_1} \}$ and $R'=\{r'_1,\dots,r'_{n_2}\}$, 
2) $X_t= \{ x_{t,1},\dots,x_{t,n_1} \}$ and $X'_t= \{ x'_{t,1},\dots,x'_{t,n_1} \}$ for all $t\le n_3$, and 
3) $Y_t= \{ y_{t,1},\dots,y_{t,n_1} \}$ and $Y'_t= \{ y'_{t,1},\dots,y'_{t,n_1} \}$ for all $t\le n_3$.
These are all vertices in the graph, so there are $\Theta(n_1+n_2+n_1n_3+n_2n_3) = \Theta(m)$ vertices in total.
Next, we add edges $(l_{i},l'_{i})$, $(r_{j},r'_{j})$, $(x_{t,i},x'_{t,i})$ and $(y_{t,j},y'_{t,j})$
for each $i\le n_1,j \le n_2, t\le n_3$. These edges form a perfect matching of size $\Theta(m)$.
Finally, we add edges $(x_{t,i},l_i)$ and $(y_{t,j},r_j)$ for each $i\le n_1,j \le n_2, t\le n_3$. 
These $\Theta(m)$-many edges do not change the size of matching. In total, there are $\Theta(m)$ edges.

Once $(u^{t},v^{t})$ arrives, we delete the edge $(x_{t,i},x'_{t,i})$ for
each $u_{i}=1$ and the edge $(y_{t,j},y'_{t,j})$ for each $v_{j}=1$.
Let $d_t$ be the number of edges we delete in this way.
Observe that $(u^{t})^{\top}Mv^{t}=1$ iff there is an edge $(l_{i},r_{j})$ for
some $i,j$ where $u_{i}=1,v_{j}=1$ iff there is an augmenting path
from $x_{t,i}$ to $y_{t,j}$ for some $i,j$. 
So if $(u^{t})^{\top}Mv^{t}=1$, then the size of maximum matching is decreased by at most $d_t-1$
Otherwise, the size of maximum matching is decreased by $d_t$.
Before $(u^{t+1},v^{t+1})$ arrives, we delete all edges $(x_{t,i},l_i)$ and $(y_{t,j},l_j)$ for each $i\le n_1,j\le n_2$.
Therefore, the graph now has a perfect matching again. 
So for each $t$, we need 1 query, and hence $n_3=\sqrt{m}$ queries in total.
\end{proof}
\begin{corr}
\label{corr:bipartitie partial}
For any $n$ and $m = O(n)$,
\Cref{oMv hard} implies that there is no partially dynamic algorithm
for bipartite maximum matching on a graph with $n$ vertices and at most $m$ edges with preprocessing time $p(n)=poly(n)$,
total update time $u(m,n)=\oo(m^{3/2})$, and query time $q(m)=\oo(m)$
that has an error probability of at most $1/3$.
 \end{corr}
\begin{proof}
	Same argument as \Cref{cor:st-SP hard partial}. 
\end{proof}
Note that the hardness proof in \Cref{corr:bipartitie partial} applies only to sparse graphs.

\paragraph{Single Source Shortest Path (ss-SP)}
\begin{lemma}
\label{lem:ss-SP reduc partial}
Given an incremental (respectively decremental) dynamic algorithm $\cA$ for ss-SP, one can solve
$(\frac{\delta}{1-\delta})$-$\ouMv$ with parameters $n_1$, $n_2$, and $n_3$
by running $\cA$ on a graph with $\Theta(m^\delta+m^{1-\delta})$ vertices
and $\Theta(m)$ edges
which is initially empty (respectively initially has $\Theta(m)$ edges), 
and then making $m^{2(1-\delta)}$ queries,
where $m$ is such that $n_1=m^{1-\delta},n_2=m^{\delta}$ and $n_3=m^{1-\delta}$.
\end{lemma}
\begin{proof}
Given an input matrix $M$ of $(\frac{\delta}{1-\delta})$-$\ouMv$, 
we construct the bipartite graph $G_{M}$, and also a path
$Q=(q_{1},q_{2},\dots,q_{n_3})$ where $q_{1}=s$. 
Add the edge $(q_{t},r_{j})$ for all $j\le n_2, t\le n_3$.
There are $\Theta(n_1+n_2+n_3) = \Theta(m^\delta+m^{1-\delta})$ vertices and
$\Theta(n_1n_2 + n_2n_3) = \Theta(m)$ edges.

Once $u^{t}$ and $v^{t}$ arrive, we disconnect $q_{t}$
from $r_{j}$ iff $v_{j}^{t}=0$. We have that if $(u^{t})^{\top}Mv^{t}=1$,
then $d(s,l_{i})=t+1$ for some $i$ where $u_{i}=1$, otherwise $d(s,l_{i})\ge t+2$
for all $i$ where $u_{i}=1$. So $n_1$ queries are enough
to distinguish these two cases. Before $u^{t+1}$ and $v^{t+1}$ arrive, we disconnect
$q_{t}$ from $r_{j}$ for all $j\le n_2$.
So for each $t$, we need $n_1$ queries, and hence $n_1n_3=m^{2(1-\delta)}$ queries in total.
\end{proof}

\paragraph{Distinguishing Distance among Vertices between 2 and 4 (ap-SP (2 vs.\ 4))}
\begin{lemma}
\label{lem:ap-SP reduc partial}
Given an incremental (respectively decremental) dynamic algorithm $\cA$ for $(\alpha,\beta)$-approximate ap-SP with $2\alpha+\beta<4$, 
one can solve $(\frac{\delta}{1-\delta})$-$\ouMv$ with parameters $n_1$, $n_2$, and $n_3$
by running $\cA$ on a graph with $\Theta(m^\delta+m^{1-\delta})$ vertices
and $\Theta(m)$ edges
which is initially empty (respectively initially has $\Theta(m)$ edges), 
and then making $m^{2(1-\delta)}$ queries,
where $m$ is such that $n_1=m^{1-\delta},n_2=m^{\delta}$ and $n_3=m^{1-\delta}$.
\end{lemma}
\begin{proof}
Given an input matrix $M$ of $(\frac{\delta}{1-\delta})$-$\ouMv$, 
we construct a bipartite graph $G_{M}$, and another set of
vertices $Q=\{q_{1},\dots,q_{n_3}\}$.
Add the edge $(q_{t},r_{j})$ for all $j\le n_2, t\le n_3$.
There are $\Theta(n_1+n_2+n_3) = \Theta(m^\delta+m^{1-\delta})$ vertices and
$\Theta(n_1n_2 + n_2n_3) = \Theta(m)$ edges.

Once $u^{t}$ and $v^{t}$ arrive, we disconnect $q_{t}$
from $r_{j}$ iff $v_{j}^{t}=0$. We have that if $(u^{t})^{\top}Mv^{t}=1$,
then $d(q_{t},l_{i})=2$ for some $i$ where $u_{i}=1$, otherwise
$d(s,l_{i})\ge4$ for all $i$ where $u_{i}=1$.
So for each $t$, we need $n_1$ queries, and hence $n_1n_3=m^{2(1-\delta)}$ queries in total.
\end{proof}

\paragraph{Transitive Closure}
\begin{lemma}
\label{lem:TC reduc partial}
Given an incremental (respectively decremental) dynamic algorithm $\cA$ for transitive closure, 
one can solve $(\frac{\delta}{1-\delta})$-$\ouMv$ with parameters  $n_1$, $n_2$, and $n_3$
by running $\cA$ on a graph with $\Theta(m^\delta+m^{1-\delta})$ vertices
and $\Theta(m)$ edges
which is initially empty (respectively initially has $\Theta(m)$ edges), 
and then making $m^{2(1-\delta)}$ queries,
where $m$ is such that $n_1=m^{1-\delta},n_2=m^{\delta}$ and $n_3=m^{1-\delta}$.
\end{lemma}
\begin{proof}
Given an input matrix $M$ of $(\frac{\delta}{1-\delta})$-$\ouMv$,  
we construct a directed graph $G_{M}$ where the edges are directed from $R$ to $L$, 
and another set of vertices $Q=\{q_{1},\dots,q_{n_3}\}$.
Add the directed edge $(q_{t},r_{j})$ for all $j\le n_2, t\le n_3$.
There are $\Theta(n_1+n_2+n_3) = \Theta(m^\delta+m^{1-\delta})$ vertices and
$\Theta(n_1n_2 + n_2n_3) = \Theta(m)$ edges.

Once $u^{t}$ and $v^{t}$ arrive, we disconnect $q_{t}$
from $r_{j}$ iff $v_{j}^{t}=0$. We have that $(u^{t})^{\top}Mv^{t}=1$
iff $q_{t}$ can reach $l_{i}$ for some $i$ where $u_{i}=1$. 
So for each $t$, we need $n_1$ queries, and hence $n_1n_3=m^{2(1-\delta)}$ queries in total. 
\end{proof}

\begin{corr}
\label{cor:trade-off hardness partial mn}
For any $n$, $m = \Theta(n^{1/(1-\delta)}\})$ and constant $\delta\in(0,1/2]$,
\Cref{oMv hard} implies that there is no partially dynamic algorithm
for the problems listed below for a graph with $n$ vertices and at most $m$ edges
with preprocessing time $p(m)=poly(m)$,
total update time $u(m)=\oo(mn)$, and query time $q(m)=\oo(m^{\delta})$ per query
that has an error probability of at most $1/3$.
The problems are:
\begin{itemize}[noitemsep,nolistsep]
\item ss-SP
\item ap-SP (2 vs.\ 4)
\item Transitive Closure
\end{itemize}
\end{corr}
\begin{proof}
Suppose there is such an algorithm $\cA$ for any problem in the list. 
By \Cref{lem:ss-SP reduc partial,lem:ap-SP reduc partial,lem:TC reduc partial}, we construct an algorithm $\cB$ for $(\frac{\delta}{1-\delta})$-$\ouMv$ with parameters $n_1=m^{1-\delta}$, $n_2=m^{\delta}$, and $n_3=m^{1-\delta}$ 
by running $\cA$ on a graph with $n_0=\Theta(m^\delta+m^{1-\delta}) = \Theta(m^{1-\delta})$ vertices and $m_0 = \Theta(m)$ edges.
Note that $m_0 = \Theta(n_0^{1/(1-\delta)})$.
Since $\cA$ uses polynomial preprocessing time,
total update time $\oo(mn)=\oo(m^{2-\delta})$
and total query time $O(m^{2(1-\delta)} q(m) )=\oo(m^{2-\delta})$,
$\cB$ has polynomial preprocessing time and $\oo(m^{2-\delta})$ computation time contradicting \Cref{oMv hard} by \Cref{ouMv hard}.
\end{proof}
	
\section{Further Discussions}\label{sec:futher discussions}

\subsection{Multiphase Problem}\label{sec:multiphase}

The multiphase problem is introduced in \cite{Patrascu10} as a problem
which can easily provide hardness for various dynamic problems. Though,
the lower bounds obtained are always worst-case time lower bounds.

In this section, we first prove that the $\oMv$ conjecture implies a tight lower bound for this problem. Then, we show a general approach for getting amortized lower bounds using the known reductions from the multiphase problem.

\begin{definition}[Multiphase Problem]
	The multiphase problem with parameters $n_1,n_2$ is a problem with 3 phases. 
	Phase 1: Preprocess a Boolean matrix $M$ of size $n_{1}\times n_{2}$ in time $O(n_{1}n_{2}\tau)$. 
	Phase 2: Get an $n_{2}$-dimensional vector $v$ and spend time $O(n_{2}\tau)$.
	Phase 3: Given an index $i$, answer if $(Mv)_{i}=1$ in time $O(\tau)$. We call $\tau$ the update time.
\end{definition}
We note that the equivalent problem definition described in \cite{Patrascu10}
uses a family of sets and a set instead of a matrix and a vector.
Assuming that there is no truly subquadratic 3SUM algorithm,  \patrascu \cite{Patrascu10}
showed that one cannot solve multiphase with parameters $n_1,n_2$ when $n_{1}=n_{2}^{2.5}$ and $\tau=\oo(\sqrt{n_{2}})$.
Based on $\oMv$, we can easily prove a better a lower bound for $\tau$. 
\begin{lemma}
	\label{lem:multiphase reduc}
	For any $\gamma>0$,
	given an algorithm for multiphase problem with parameters $n_1,n_2$ and $\tau$ update time, 
	one can solve $\goMv$ with parameters $n_1,n_2,n_3$ in
	time $O((n_{1}n_{2}+n_{2}n_{3}+n_{1}n_{3})\tau)$.\end{lemma}
\begin{proof}
	Given a matrix $M$ of $n_{1}\times n_{2}$, we run Phase 1 of the multiphase
	algorithm in time $O(n_{1}n_{2}\tau)$. For every $v^{t}$, $1\le t\le n_{3}$,
	we run $n_{3}$ many instances of Phase 2 in time $n_{3}\times O(n_{2}\tau)$.
	To compute $Mv^{t}$, we run $n_{1}n_{3}$ many instances of Phase
	3 in time $n_{1}n_{3}\times O(\tau)$.\end{proof}
\begin{corollary}
	\label{corr:multiphase}
	For any $n_1,n_2$, \Cref{oMv hard} implies that there is no algorithm $\cA$ for the multiphase problem with parameters $n_1,n_2$ such that  $\tau=\oo(\min\{n_{1},n_{2}\})$.
\end{corollary}
\begin{proof}
	Suppose there is such an algorithm $\cA$. Then, by \Cref{lem:multiphase reduc} and setting $n_3 = \min\{n_1,n_2\}$, one can solve
	$\goMv$ with parameters $n_1,n_2,n_3$ in time $O(n_{1}n_{2}+n_{2}n_{3}+n_{1}n_{3})\times \oo(\min\{n_{1},n_{2}\})$. 
	Assume w.l.o.g. that $n_1 \le n_2$ then we get the expression $O(n_1 n_2 + n_1^2) \times \oo(n_1) = \oo(n_1^2 n_2) = \oo(n_1 n_2 n_3)$, which contradicts \Cref{oMv hard} by \Cref{general oMv hard}.
\end{proof}

\subsubsection{Converting Worst-case Bounds to Amortized Bounds}
In the following, let $\cA$ be a fully dynamic algorithm that maintains some object $G$, e.g., a graph, a matrix, an array etc.
Similar to \Cref{lem:multiphase reduc}, by running many instances of phase 2 and 3 of the multiphase algorithm, we have the following.
\begin{lemma}\label{lem:multiphase reduc2}
	For any constant $\gamma>0$, suppose that one can solve the multiphase problem with parameters $n_1,n_2$ by running $\cA$ on $G$ of size $s(n_1,n_2)$ using 
	$p(n_1,n_2)$ preprocessing time, and then making $k_i$ updates/queries on phase $i$, for $i \in {2,3}$. 
	Then one can solve $\goMv$ problems with parameters $n_1,n_2,n_3$ by running $\cA$ on $G$ using $p(n_1,n_2)$ preprocessing time, and then using $O(k_2 n_3 + k_3 n_1 n_3)$ updates/queries. 
\end{lemma}

\begin{corollary}
	\label{amortized from multiphase}
	Suppose that one can solve the multiphase problem with parameters $n_1,n_2$, where $n_1=n_2^\gamma$, by running $\cA$ on $G$ of size $s(n_1,n_2)$ using $\poly(n_1,n_2)$ preprocessing time, and then using $k_i$ updates/queries on phase $i$, for $i \in {2,3}$.
	Then \Cref{oMv hard} implies that $\cA$ cannot maintain an object $G$ of size $s(n_1,n_2)$ with polynomial preprocessing time and $\oo(\min\{ \frac{n_1 n_2}{k_2}  , \frac{n_2}{k_3}  \})$ amortized update and query time.
\end{corollary}

\begin{proof}
	Suppose that $\cA$ has such an amortized update/query time. Then, by \Cref{lem:multiphase reduc2}, one can solve $\goMv$ with parameters $n_1,n_2,n_3$ using $\poly(n_1,n_2)$ preprocessing time, and computation time 
	$\oo(\min\{ \frac{n_1 n_2}{k_2}  , \frac{n_2}{k_3}  \} \times O( k_2 n_3 + k_3 n_1 n_3) = \oo(n_1 n_2 n_3)$. This contradicts \Cref{oMv hard} by \Cref{general oMv hard}.
\end{proof}

\paragraph{Example.}
It is shown in \cite{Patrascu10} that, given a fully dynamic algorithm
$\cA$ for subgraph connectivity that runs on a graph with $O(n_{1}+n_{2})$ vertices,
the multiphase problem with parameters $n_1,n_2$ can be solved using $k_{2}=n_{2}$ operations in
phase 2 and $k_{3}=1$ operations in phase 3. 
By \Cref{amortized from multiphase} and by setting $n_{1}=n_{2}=n$, the amortized cost of $\cA$ on a graph with $\Theta(n)$ vertices
cannot be $\oo(n)$ unless the $\oMv$ conjecture fails.
Note however that this lower bound is subsumed by the one we give in \Cref{cor:high query}.

\subsection{Open Problems}\label{sec:open}

Of course it is very interesting to settle the \oMv conjecture. 
Besides this, there are still many problems for which this work does not provide tight lower bounds, and it is interesting to prove such lower bounds based on the \oMv or other reasonable conjectures.

\paragraph{Minimum Cut.} 
Thorup and Karger \cite{ThorupK00} presented a $(2 + o(1))$-approximation algorithm with polylogarithmic amortized update time. Thorup \cite{Thorup07mincut} showed that in $\tilde O(\sqrt{n})$ worst-case update time the minimum cut can be maintained exactly when the minimum cut size is small and $(1+\epsilon)$-approximately otherwise. Improving this result using amortization is mentioned as a major open problem in \cite{Thorup07mincut}. 
Very recently Fakcharoenphol~et~al. \cite{FakcharoenpholKNSS14} showed some related hardness results (e.g., for some subroutine used in Thorup's algorithm). 
However, currently there is no evidence that the minimum cut cannot be maintained in polylogarithmic update time. 
In fact, it is not even known if polylogarithmic update time is possible or impossible for a key subroutine in Thorup's algorithm called {\em min-tree cut}, where we are given edge updates on a graph and its spanning tree and have to maintain the minimum cut among the cuts obtained by removing one edge from the spanning tree. We believe that understanding this subroutine is an important step in understanding the dynamic minimum cut problem.

\paragraph{Approximation Algorithms for Non-Distance Problems.} In this paper we provide hardness results for several approximation algorithms for distance-related problems. However, we could not extend the techniques to non-distance graph problems in undirected graphs such as approximating maximum matching, minimum cut, maximum flow, and maximum densest subgraph.

\paragraph{Total Update Time for Partially Dynamic Algorithms.} While there are many hardness results in the partially dynamic setting in this and previous work, quite a few problems are still open. 
Of particular interest are problems that are known to be easy when one type of updates is allowed but become challenging when the other type of updates is allowed. 
For example, the single-source reachability problem can be solved in $O(m)$ total update time in the incremental setting \cite{Italiano86} but the best time in the decremental setting is still larger than $mn^{0.9}$ (e.g., \cite{HenzingerKNstoc14,HenzingerKN-ICALP15}). This is also the case for the minimum cut problem where the incremental setting can be $(1+\epsilon)$-approximated in $\tilde O(m)$ total update time \cite{Henzinger97} while the current best decremental $(1+\epsilon)$-approximation algorithm requires $\tilde O(m+n\sqrt{n})$ total update time \cite{Thorup07mincut}, and the topological ordering problem which is trivial in the decremental setting but challenging otherwise (e.g., \cite{HaeuplerKMST12,BenderFG09}).
Since known hardness techniques -- including those presented in this paper -- usually work for both the incremental and the decremental setting, proving non-trivial hardness results for the above problems seems to be challenging.

\paragraph{Worst-Case Update Time.} While there are fully dynamic algorithms with polylogarithmic {\em amortized} update time for many problems, not much is known for {\em worst-case} update time. The only exception that we are aware of is the connectivity problem (due to the recent breakthrough of Kapron~et~al.~\cite{KapronKM13}). Other basic graph problems, such as minimum spanning tree, 2-edge connectivity, biconnectivity, maximum matching, and densest subgraph are not known to have polylogarithmic worst-case update time. A polynomial hardness result for worst-case update time for these problems based on a natural assumption will be very interesting. 
The challenge in obtaining this is that such a result must hold only for the worst-case update time and not for the amortized one. Such results were published previously (e.g., those in \cite{AbboudW14,Patrascu10,KopelowitzPP14}), but most of these results are now known to hold for amortized update time as well assuming \oMv and SETH (some exceptions are those for partially dynamic problems).

\paragraph{Deterministic Algorithms.} Derandomizing the current best randomized algorithms is an important question for many problems, e.g., approximate decremental single-source and all-pairs shortest paths \cite{Bernstein13,HenzingerKNFOCS14} and worst-case connectivity and spanning tree \cite{KapronKM13}.
This is important since deterministic algorithms do not have to limit the power of the adversary generating the sequence of updates and queries.
Proving that derandomization is impossible for some problems will be very interesting. The challenge is that such hardness results must hold only for deterministic algorithm and not for randomized algorithms.

\paragraph{Trade-off between Query and Update Time} In this paper we present hardness results with a trade-off between query and update time for several problems. Are these hardness results tight? This seems to be a non-trivial question since not much is known about the upper bounds for these problems. A problem for which it is particularly interesting to study this question is the subgraph connectivity problem, since it is the starting point of many reductions that lead to hardness results with a trade-off.
In this paper, we show that for any $ 0 < \alpha < 1 $ getting an $O(m^\alpha)$ update time requires a query time of $\Omega(m^{1-\alpha})$. This matches two known upper bounds in \cite{Duan10,ChanPR11} when $\alpha=4/5$ and $\alpha=2/3$. It is reasonable to conjecture that there is a matching upper bound for all $ 0 < \alpha < 1 $; however, it is not clear if this is true or not.

\section{Acknowledgement}
D.~Nanongkai would like to thank Parinya Chalermsoomsook, Bundit Laekhanukit, and Virginia Vassilevska Williams for comments. T.~Saranurak would like to thank Markus Bl\"aser and Gorav Jindal for discussions. 

\printbibliography[heading=bibintoc] 

\part*{Appendix}
\appendix
\section{Conjectures (from \cite{AbboudW14})}\label{sec:conjectures}

\begin{conjecture}
[No truly subquadratic 3SUM (3SUM)]In the Word RAM model with words of $O(\log n)$
bits, any algorithm requires $n^{2-o(1)}$ time in expectation to
determine whether a set $S\subset\{-n^{3},\dots,n^{3}\}$ of $|S|=n$
integers contains three distinct elements $a,b,c\in S$ with $a+b=c$.
\end{conjecture}

\begin{conjecture}
[No truly subcubic APSP (APSP)] There is a constant $c$, such that in the
Word RAM model with words of $O(\log n)$ bits, any algorithm requires
$n^{3-o(1)}$ time in expectation to compute the distances between
every pair of vertices in an $n$ node graph with edge weights in
$\{1,...,n^{c}\}$.
\end{conjecture}

\begin{conjecture}
[Strong Exponential Time Hypothesis (SETH)] For every $\epsilon>0$,
there exists a $k$, such that SAT on $k$-CNF formulas on $n$ variables
cannot be solved in $O^{*}(2^{(1-\epsilon)n})$ time, where $O^{*}(\cdot)$
hides polynomial factor.
\end{conjecture}

\begin{conjecture}
[No almost linear time triangle (Triangle)] There is a constant $\delta$ > 0,
such that in the Word RAM model with words of $O(\log n)$ bits, any
algorithm requires $m^{1+\delta-o(1)}$ time in expectation to detect
whether an $m$ edge graph contains a triangle.
\end{conjecture}

\begin{conjecture}
[No truly subcubic combinatorial BMM (BMM)] In the Word RAM model with words
of $O(\log n)$ bits, any \emph{combinatorial} algorithm requires
$n^{3-o(1)}$ time in expectation to compute the Boolean product of
two $n\times n$ matrices. \end{conjecture}

\end{document}